\documentclass[10pt]{article}
\usepackage{amsmath,amssymb,amsthm,amsfonts,amsbsy,bm,bbm,latexsym,color}
\usepackage[letterpaper,margin=1.2in]{geometry}
\usepackage[english]{babel}
\usepackage{hyperref}
\usepackage{mathrsfs}
\usepackage{mathabx}

\newtheorem{theorem}{Theorem}[section]
\newtheorem{lemma}[theorem]{Lemma}
\newtheorem{assumption}[theorem]{Assumption}

\newtheorem{remark}[theorem]{Remark}

\renewcommand{\Re}{\mathrm{Re}}
\renewcommand{\Im}{\mathrm{Im}}
\renewcommand{\d}{\mathop{}\!\mathrm{d}}

\newcommand{\hc}{\mathrm{h.c.}}

\newcommand{\tr}{\mathrm{tr}}

\newcommand{\id}{\mathbbm{1}}

\newcommand{\Hilbert}{\mathcal{H}}

\newcommand{\Fock}{\mathcal{F}}
\newcommand{\Gock}{\mathcal{G}}

\newcommand{\Number}{\mathcal{N}}

\newcommand{\RRR}{\mathbb{R}}

\newcommand{\NNN}{\mathbb{N}}

\newcommand{\norm}[1]{\left\| #1 \right\|}
\newcommand{\scp}[2]{\big\langle #1 , #2 \big\rangle}
\newcommand{\SCP}[2]{\Big\langle #1 , #2 \Big\rangle}
\newcommand{\bra}[1]{\langle #1 |}
\newcommand{\ket}[1]{| #1 \rangle}
\newcommand{\ketbra}[2]{| #1 \rangle \langle #2 |}
\newcommand{\ketbr}[1]{| #1 \rangle \langle #1 |}

\newcommand{\gammapart}{\mu_{\Psi_N(t)}^{\mathrm{part}}}
\newcommand{\gammapartFock}{\gamma_{\chi_{\leq N}(t)}^{\mathrm{part}}}
\newcommand{\gammapartFockPhi}{\gamma_{\Phi}^{\mathrm{part}}}
\newcommand{\gammapartFockz}{\gamma_{\chi_{0}(t)}^{\mathrm{part}}}
\newcommand{\betapart}{\beta_{\chi_{\leq N}(t)}^{\mathrm{part}}}
\newcommand{\betapartPhi}{\beta_{\Phi}^{\mathrm{part}}}
\newcommand{\betapartzz}{\beta_{00}^{\mathrm{part}}}
\newcommand{\betapartzo}{\beta_{01}^{\mathrm{part}}}

\newcommand{\gammafield}{\mu_{\Psi_N(t)}^{\mathrm{field}}}
\newcommand{\gammafieldFock}{\gamma_{\chi_{\leq N}(t)}^{\mathrm{field}}}
\newcommand{\gammafieldFockPhi}{\gamma_{\Phi}^{\mathrm{field}}}
\newcommand{\gammafieldFockz}{\gamma_{\chi_{0}(t)}^{\mathrm{field}}}
\newcommand{\betafield}{\beta_{\chi_{\leq N}(t)}^{\mathrm{field}}}
\newcommand{\betafieldPhi}{\beta_{\Phi}^{\mathrm{field}}}

\newcommand{\betafieldzo}{\beta_{01}^{\mathrm{field}}}

\allowdisplaybreaks

\begin{document}

\title{Bogoliubov Dynamics and Higher-order Corrections for the Regularized Nelson Model}

\author{
Marco Falconi\thanks{Dipartimento di Matematica, Politecnico di Milano, Piazza Leonardo da Vinci 32, 20133 Milano, Italy. Email: \texttt{marco.falconi@polimi.it}}, 
Nikolai Leopold\thanks{University of Basel, Department of Mathematics and Computer Science, Spiegelgasse 1, 4051 Basel,	Switzerland. Email: \texttt{nikolai.leopold@unibas.ch}},
David Mitrouskas\thanks{Institute of Science and Technology (IST) Austria, Am Campus 1, 3400 Klosterneuburg, Austria. Email: \texttt{david.mitrouskas@ist.ac.at}}, and 
Sören Petrat\thanks{Department of Mathematics and Logistics, Jacobs University Bremen, Campus Ring 1, 28759 Bremen, Germany. Email: \texttt{s.petrat@jacobs-university.de}}
}

\frenchspacing

\maketitle

\begin{abstract}
We study the time evolution of the Nelson model in a mean-field limit in which $N$ non-relativistic bosons weakly couple (w.r.t. the particle number) to a positive or zero mass quantized scalar field. Our main result is the derivation of the Bogoliubov dynamics and higher-order corrections. More precisely, we prove the convergence of the approximate wave function to the many-body wave function in norm, with a convergence rate proportional to the number of corrections taken into account in the approximation. We prove an analogous result for the unitary propagator. As an application, we derive a simple system of PDEs describing the time evolution of the first- and second-order approximation to the one-particle reduced density matrices of the particles and the quantum field, respectively.
\end{abstract}

\noindent
\textbf{MSC class:} 35Q40, 35Q55, 81Q05, 81T10, 81V73, 82C10

\section{Introduction}
\label{sec:introduction}

The Nelson model describes the interaction between a relativistic scalar field and nonrelativistic particles. Since its introduction in the
community of mathematical physics by E.\ Nelson \cite{nelson} it has served as a playground for a rigorous study of many
features of quantum field theory, in a simple but nonetheless realistic
model. It was originally introduced to study the interaction of nonrelativistic spinless nucleons with a scalar meson field  but has also been used in condensed matter physics and quantum optics to describe either particles in
interaction with phonons in a crystal, or particles interacting with
radiation (in a linearized approximation). In this paper, we focus on the
regular Nelson Hamiltonian, whose interaction is cut off for high momenta by
an ultraviolet regularization, as this analysis serves as a starting point
for the more interesting case of a renormalized interaction without
ultraviolet regularization.

The Nelson model has been studied from different points of view (see,
e.g., \cite{A2000, DG1999, GW2018, P2005, M2006, M2006_2, LSTT2018} and
references therein). Here, we consider $N\gg 1$ nonrelativistic particles coupled to the quantum field, and focus on its effective description in a regime in which the nonrelativistic particles are close to a Bose--Einstein
condensate and the scalar field is macroscopic, and thus behaving classically (with $N^{-1/2}$ quantifying the ``degree of quantumness'', as a semiclassical $\hslash$
parameter). This regime has already been studied using semiclassical
techniques \cite{AF2014,AF2017}, Fock space methods \cite{falconi} or a
combination of the coherent state approach and the method of counting
\cite{LP2018}.  In \cite{davies, hiroshima1998, teufel} it has, moreover,
been shown that the effect of the quantum field on the particles can in
certain situations be approximated by a direct pair interaction. Let us also mention further works that focus on related
models \cite{LMS2021, LP2019, LP2020}, a partial limit \cite{CCFO2019, CF2018, CFO2019, CFO2020} and
the strong coupling limit of the polaron \cite{FRS2021, FG2017, FS2014, G2017, LMRSS2021, LRSS2019, M2021}.

In this paper we take a different
route, drawing inspiration by a series of works by several authors
\cite{Boccato:2016, brennecke:2017, Chong2016, GV1979a, GV1979b,Ginbre_Velo_expansion2,Ginbre_Velo_expansion1,
  GrillakisMachedon:2013, GM2017, GMM2010, GMM2011, Kuz2017, Lewin:2015a,
  mpp, nam:2016, nam:2015, namnap_review, namnap_low_dim, soffer, RS2009} on systems
of many nonrelativistic bosons in direct pair interaction. In particular, we
follow the route from \cite{bossmann} (see also \cite{BPS2021} for a related
result in the static setting). We construct a Bogoliubov theory, and next-order
corrections, for the Nelson Hamiltonian in the mean-field/semiclassical
regime described above. In addition to the inherent interest of building up
such a theory for this model, this also allows to strengthen the previously
available propagation of chaos results from weak-* to strong convergence, at
least for a suitable class of initial states: the wave functions converge, at
any time and with an explicit rate of convergence, in Hilbert space norm and
not only through the expectation of suitable observables (of course, provided
that strong convergence holds at a given initial time).
Moreover, we provide a series expansion that approximates the two-parameter group generated by the Nelson Hamiltonian on a subspace of the excitation Fock space in the strong sense to arbitrary precision.
As an application we show that the time evolution of the first and second-order approximation of one-particle and one-field boson reduced density matrices are determined by a self-contained system of PDEs.

The article is organized as follows. In the rest of this section we
introduce the main notations and mathematical definitions used throughout the
paper, providing a more detailed introduction to the problem at hand. In Section \ref{sec:main-results}, we state the main convergence results of
this paper, viz., Theorems \ref{thm_Bog}, \ref{thm_higher_orders}, and
\ref{thm_higher_orders_unitary_group}. In Section
\ref{sec:discussion-results} we discuss the usefulness of our results by explaining how Bogoliubov theory can be used for computations, and by deriving the next-order PDEs for the reduced density matrices. Section \ref{sec_Bog_eq} discusses the Bogoliubov dynamics,
and its well-posedness, following mainly \cite{falconi, Lewin:2015a, nam:2016}. Finally, in Section \ref{sec:proofs} we provide the technical
details that lead to the proof of Theorems \ref{thm_Bog}, \ref{thm_higher_orders}, and \ref{thm_higher_orders_unitary_group}.

We assume that the reader is familiar with the basic objects of Fock spaces such as creation and annihilation operators, the number operator, and sectors with a fixed
number of particles. For an introduction to the topic we refer to \cite{DG2013}.

\subsection{Basic definitions and notations}
\label{sec:basic-defin-notat}

The Nelson model we are considering is set up as follows. Let us consider a
fixed number $N$ of nonrelativistic bosons, coupled to a quantized scalar
field with an ultraviolet-regular interaction. The Hilbert space of such a
theory is
\begin{equation}
\Hilbert_N = (L^2(\mathbb{R}^3))^{\otimes_s N} \otimes \Fock,
\end{equation}
where $\Fock$ is the bosonic Fock space over $L^2(\mathbb{R}^3)$, and $\otimes_s$ is the symmetric tensor product.
Any wave function $\Psi_N(t)$, where we have made explicit
the dependence on the number of nonrelativistic particles $N$ that we are going to choose very large, evolves with the Schr\"odinger equation
\begin{equation}\label{Schr_eq}
i\partial_t \Psi_N(t) = H_N^{\mathrm{Nelson}} \Psi_N(t),
\end{equation}
generated by the cutoff  Nelson Hamiltonian with mean-field scaling,
\begin{equation}
H_N^{\mathrm{Nelson}} = - \sum_{j=1}^N \Delta_j + \int dk\, \omega(k) a^{*}(k) a(k) + N^{-1/2} \sum_{j=1}^N \widehat{\Phi}(x_j).
\end{equation}
Here, $a^*$ and $a$ are the bosonic creation and annihilation operators satisfying the canonical commutation relations
\begin{equation}
[a(k),a^*(\ell)] = \delta(k-\ell), ~~ [a(k),a(\ell)] = 0 = [a^*(k),a^*(\ell)],
\end{equation}
and the field operator is defined as
\begin{equation}
\widehat{\Phi}(x) = \int dk\, \eta(k)e^{- 2\pi  ikx} \Big( a^*(k) + a(-k) \Big)
\end{equation}
with
\begin{equation}
\eta(k) = \frac{g(k)}{\sqrt{2\omega(k)}},
\end{equation}
for some real cutoff function $g$, satisfying $g(k) = g(-k)$\footnote{The choice of a real cutoff
	function $g$ is taken only to simplify some formulas in the
	presentation. The generalization to an arbitrary complex $g$ is straightforward.} and $\eta \in L^2$, with dispersion
$\omega(k) = \sqrt{k^2 + m^2}$, $m\geq 0$. We also assume $\frac{\eta}{\sqrt{\omega}} \in L^2$, which makes $H_N^{\mathrm{Nelson}}$ self-adjoint on $D(H_N^{\mathrm{Nelson},0})$, where $H_N^{\mathrm{Nelson},0} = - \sum_{j=1}^N \Delta_j + \int dk\, \omega(k) a^{*}(k) a(k)$ is the non-interacting part of the Nelson Hamiltonian.

In order to approximate the time evolution of the Nelson model we introduce the classical field
\begin{equation}
\label{Phi_eq}
\Phi(t,x) = \int dk\, \eta(k) e^{-2\pi ikx} \Big( \overline{\alpha(t,k)} + \alpha(t,-k)\Big)
\end{equation}
and the effective equations
\begin{subequations}
	\begin{align}
	\label{varphi_eq2} 
	i \partial_tu(t) &= \big( -\Delta + \Phi(t,\cdot) \big) u(t),\\
	\label{alpha_eq2} i \partial_t \alpha(t) &= \omega \alpha(t) + \eta \widehat{|u(t)|^2},
	\end{align}
\end{subequations}
with $u(t), \alpha(t) \in L^2(\RRR^3)$ for all $t$ (and where $\widehat{|u(t)|^2}$ denotes the Fourier transform of $|u(t)|^2$). Here, $\alpha$ and $\bar{\alpha}$ should be considered as the classical counterparts of $a$ and $a^{*}$. These equations are called Schr\"odinger--Klein--Gordon equations, since taking a second time derivative of Equation~\eqref{alpha_eq2} yields
\begin{equation}
(\partial_t^2 - \Delta + m^2 ) \Phi(t) = - \widecheck{g^2} * \lvert u(t)  \rvert_{}^2 \;,
\end{equation}
a Klein--Gordon equation with source term ($\widecheck{g^2}$ denotes the inverse Fourier transform of $g^2$).

Given a solution
$\bigl(u(t),\alpha(t)\bigr)$ of the Schr\"odinger--Klein--Gordon equations, let us
define the wave function $\varphi(t)$ by
\begin{equation}
  \label{eq:9}
  \varphi(t,x)= e^{i \mu(t)}u(t,x)\;,
\end{equation}
with the time-dependent phase $\mu(t)$ defined by
\begin{equation}
\mu(t) := \frac{1}{2} \int dx\, \Phi(t,x) |u(t,x)|^2\;.
\end{equation}

The Schr\"odinger--Klein--Gordon system is globally well-posed in several
function spaces. In this work, we rely on the following well-posedness
lemma. For $\ell\in\mathbb{N}$, let $H^{\ell}(\mathbb{R}^3)$ be the Sobolev space of order $\ell$ and
$L_{\ell}^2(\mathbb{R}^3)$ be a weighted $L^2$-space with norm $\| \alpha \|_{L_{\ell}^2(\mathbb{R}^3)} =
\big\| \left( 1 + | \cdot |^2 \right)^{\ell/2}
\big\|_{L^2(\mathbb{R}^3)}$
\begin{lemma}
\label{lemma: solution theory}
Let $(u_0, \alpha_0) \in H^2(\mathbb{R}^3) \times L_1^2(\mathbb{R}^3)$. Then there is a continuous map $t\mapsto
(u(t), \alpha(t))$ from $\RRR$ to $H^2(\mathbb{R}^3) \oplus L_1^2(\mathbb{R}^3)$ that satisfies
\eqref{varphi_eq2}-\eqref{alpha_eq2} with initial condition $(u(t), \alpha(t)) |_{t=0} = (u_0, \alpha_0)$.
\end{lemma}

In \cite[Appendix B ($N= \delta= 1$)]{LP2019}, Lemma \ref{lemma: solution theory}
was proved for $g(k)=
\id_{|k| \leq \Lambda}(k)$ with $\Lambda >0$. The proof, however, can easily be extended to
a larger class of cutoff functions. We remark that a similar result was shown
in \cite{falconi}, and we refer the interested reader to \cite{ CHT2018,
  pecher} for a dedicated analysis of well-posedness for the
Schr\"odinger--Klein--Gordon system without ultraviolet
cutoff.

For convenience, let us denote
\begin{align}
h(t) =  -\Delta + \Phi(t,\cdot) - \mu(t),
\end{align}
and remark that given any solution $\bigl(u(t),\alpha(t)\bigr)$ of
\eqref{varphi_eq2}-\eqref{alpha_eq2}, the couple $\bigl(\varphi(t),\alpha(t)\bigr)$ is a
solution of
\begin{subequations}
\begin{align}
\label{varphi_eq} i \partial_t\varphi(t) &=  h(t)\varphi(t) \;,\\
\label{alpha_eq} i \partial_t \alpha(t) &= \omega \alpha(t) + \eta \widehat{|\varphi(t)|^2}\; .
\end{align}
\end{subequations}
Moreover, let us define the Weyl operator
\begin{equation}
W(f) = \exp \bigg( \int dk\, \Big( f(k)a^{*}(k) - \overline{f(k)} a(k) \Big) \bigg)
\end{equation}
whose action on the vacuum of $\mathcal{F}$ creates a coherent state with mean particle number $\| f \|_{L^2(\mathbb{R}^3)}^2$ (see, e.g., \cite{BR1981}  for a more detailed introduction). 

In this work, we are interested in the evolution of a Bose--Einstein condensate of particles with initial condensate wave function $\varphi_0$ and a coherent state $W(\sqrt{N}\alpha_0)\lvert\Omega\rangle$ of field bosons with mean particle number $N \| \alpha_0 \|_{L^2(\mathbb{R}^3)}^2$ (here, $\lvert\Omega\rangle$ denotes the vacuum vector). We will prove the persistence of the condensate and the coherent structure of the field during the time evolution and show that they are described by $(\varphi(t), \alpha(t))$ evolving according to \eqref{varphi_eq} and \eqref{alpha_eq} with initial condition $(\varphi(t), \alpha(t)) \big\rvert_{t=0} = (\varphi_0, \alpha_0)$. Moreover, we will give an explicit description of the fluctuations around the Schr\"odinger--Klein--Gordon equations.
For this purpose it is convenient to embed the Hilbert space of the particles into a second Fock space, to factor out the condensate as well as the coherent state and to look at the corresponding excitation Fock spaces (for a similar strategy without second quantization see \cite{pickl, mpp}).

The excitation Fock space of the particles is given by
\begin{align}\label{eq: definition excitation Fock space particles}
\Fock_{ b }  = \bigoplus_{k=0}^\infty \mathcal F_{  b}^{(k)}  \quad \text{with}\quad \Fock_{  b}^{(k)} =  ( \varphi(t)^\perp)^{\otimes_s k}.
\end{align}
It describes bosonic particles, each with a wave function
orthogonal to the reference state $\varphi(t)$. The excitation space for the
phonons is defined in a slightly different fashion as
\begin{equation}
\label{eq: definition excitation Fock space field bosons} 
  \Fock_{a}  =  W^* \big(\sqrt{N}\alpha(t)\big) \Fock\; .
\end{equation}
This space is a unitary ``rotation'' of the original space; however, it is
the coherent state $W\big(\sqrt{N}\alpha(t)\big)\lvert\Omega\rangle$, $\lvert\Omega\rangle$ being the Fock vacuum
vector, that now plays the role of a new vacuum $\lvert\Omega_a\rangle$ for $\Fock_{a}$. The
combined excitation space is then given by the double Fock space
\begin{align}
\label{eq:double Fock space}
\mathcal G = \Fock_{b}  \otimes  \Fock_{a}.
\end{align}
Let us remark that, even though we omit this in our notation, both spaces
$\Fock_{ b }$ and $\Fock_a$ depend on time. Let us denote the creation and annihilation operators on $\mathcal F_b $ by
$b^*(x)$ and $b(x)$, and on $\mathcal F_a$, as before, by $a^{*}(k)$ and $a(k)$. We call the respective number operators $\Number_b$ and
$\Number_a$. We also introduce the operator that counts the total number of excitations,
\begin{align}
\Number = \Number_a  + \Number_b.
\end{align}

We now consider the decomposition
\begin{equation}
\label{decomp_Psi}
\Psi_N(t) = W\big(\sqrt{N}\alpha(t)\big) \sum_{k=0}^N \varphi(t)^{\otimes (N-k)} \otimes_s \chi_{\le N}^{(k)}(t) 
\end{equation}
where $\chi_{\le N}^{(k)}(t) \in \Fock_b^{(k)} \otimes \Fock_a $. For given $(\varphi(t),\alpha(t))$, this
establishes a unitary map\footnote{The unitary map and some of its properties are discussed in greater detail in Appendix \ref{section: more details on the excitation Fock space}. } between $\mathcal H_N$ and the $N$-particle excitation
space
\begin{align}\label{eq: definition k-particles excitation space}
\mathcal G_{\leq N} := \bigg( \bigoplus_{k=0}^N \mathcal F_{  b}^{(k)}  \bigg) \otimes \mathcal F_a .
\end{align}
The inverse of \eqref{decomp_Psi} is
\begin{equation}\label{psi_to_excitations}
\chi_{\le N}^{(k)}(t) = \binom{N}{k}^{1/2} \prod_{i=1}^k q_i(t) \scp{ \varphi(t)^{\otimes (N-k)} }{ W^{*}\big(\sqrt{N}\alpha(t)\big) \Psi_N(t)}_{L^2(\RRR^{d(N-k)})},\quad k = 0, 1,\ldots, N,
\end{equation}
where we take a partial inner product w.r.t. the coordinates $x_{k+1},\ldots,
x_N$, and $q_i(t) = 1 - p_i(t)$ with $p_i(t) = | \varphi(t) \rangle \langle \varphi(t) |_i$, that is,
the projector onto the state $\varphi(t)$ in the $x_i$ coordinate. We can thus
equally express the Schr\"odinger equation \eqref{Schr_eq} as
\begin{equation}\label{Schr_eq_Fock}
i\partial_t \chi_{\le N}(t) = H^{\leq N}(t) \chi_{\le N }(t),
\end{equation}
where
\begin{equation}\label{chi_in_Fock}
\chi_{\le N}(t)  = \big (\chi_{\le N}^{(k)}(t) \big)_{k=0}^{N} \, \in \Gock_{\le N}.
\end{equation}
The Hamiltonian $H^{\leq N}(t)$ can be written (see Appendix \ref{section: more details on the excitation Fock space}) as the restriction
of a Hamiltonian $H(t)$ (defined on $\Gock$) to $\mathcal G_{\leq N} $, i.e., $H^{\leq N}(t) = H(t)
|_{\Gock_{\leq N}}$, with\footnote{For one-body operators on $L^2(\mathbb{R}^3)$ with kernel $A(x,y)$ we use the usual shorthand notation $\int dx \, b^{*}(x) A b(x) = \int \int dx \, dy \, b^{*}(x) A(x,y) b(y)$.}
\begin{align}
\label{full_H}
\begin{split}
H(t) &= \int dx\, b^{*}(x) h(t) b(x) + \int dk\, \omega(k) a^{*}(k) a(k) \\
&\quad + \int dx \int dk\,  K(t,k,x) \big( a^{*}(k) + a(-k)\big) b^*(x) \big[ 1- N^{-1}\Number_b \big]_+^{1/2} + \text{h.c.}  \\
&\quad + N^{-1/2} \int dx\, b^{*}(x) \Big( q(t) \widehat{\Phi} q(t) - \scp{\varphi(t)}{\widehat{\Phi} \varphi(t)} \Big) b(x),
\end{split}
\end{align}
where $[x]_+$ denotes the positive part of $x$, $\text{h.c.}$ denotes the Hermitian conjugate of the preceding term, and $q(t)$ is the operator with
integral kernel
\begin{align}\label{Kernel_q}
q(t,x,y) = \delta(x-y) - \varphi(t,x) \overline{\varphi(t,y)} .
\end{align}
Moreover,
\begin{align}
K(t,k,x) &= \int dy \, q(t,x, y) \widetilde{K}(t,k,y) \quad \text{with}\quad
 \widetilde{K}(t,k,x) = \eta(k) e^{-2\pi ikx} \varphi(t,x),
\end{align}
and we denote by $K(t)$ the operator with operator kernel $K(t,k,x)$, i.e., $K(t) = \widetilde K(t)\overline{q(t)}$ where $\widetilde K(t)$ has kernel $\widetilde
K(t,k,x)$ and $q(t)$ has the kernel given by \eqref{Kernel_q}. Note that if an operator $A$ has integral kernel $A(x,y)$, we write $\overline{A}$ for the operator with integral kernel $\overline{A(x,y)}$.

The Bogoliubov approximation is to disregard all terms with more than two
$a^{*}$, $a$, $b^{*}$ or $b$ operators in \eqref{full_H}. This leads us to define the
Bogoliubov Hamiltonian
\begin{align}\label{H_Bog}
\begin{split}
H_0 (t) &= \int dx\, b^{*}(x) h(t) b(x) + \int dk\, \omega(k) a^{*}(k) a(k) \\
&\quad + \int dx \int dk\,  K(t,k,x) \big( a^{*}(k) + a(-k) \big) b^*(x) + \text{h.c.} \, .
\end{split}
\end{align}
While the Hamiltonian $H(t)$ maps $\Gock_{\leq N}$ to itself (because of the appearance of the square root in the second line of \eqref{full_H}) this does not hold for $H_0(t)$. The corresponding time evolution, usually referred to as Bogoliubov equation, must therefore be defined on the double  Fock space \eqref{eq:double Fock space}. It reads
\begin{equation}\label{Bog_Schr}
i\partial_t \chi_0(t) = H_0 (t) \chi_0 (t) ,
\end{equation}
where $\chi_0(t) \in \mathcal G$. Its well-posedness is discussed in
Section~\ref{sec_Bog_eq}.

Let us remark that for states $\chi$ in $\mathcal G$ and $\mathcal G_{\le N}$ we shall always use the notation $\chi^{(k)}$ to indicate the component in the $k$ particle sector w.r.t. the particle excitations, i.e.,
\begin{align}
\chi^{(k)} \in \mathcal F_b^{(k)} \otimes  \mathcal F_a.
\end{align}
In particular, every $\chi\in \Gock$ corresponds to a sequence $(\chi^{(k)})_{k\ge 0}$.

\section{Main Results}
\label{sec:main-results}

In this section we state the main results of this paper. We start with
Bogoliubov theory, and the first result on norm convergence. Then, we focus
on higher-order corrections, and the refinement of the rate of convergence for initial states that admit a power series expansion in the parameter $N\gg 1$. As our last result, we provide a similar statement for the unitary propagator.

\subsection{Bogoliubov Theory}

Our first result shows that Bogoliubov theory approximates the microscopic dynamics well, up to an error of order $N^{-1/2}$.

\begin{theorem}\label{thm_Bog}
Let $\Psi_N(t)$ be the solution to the Schr\"odinger equation \eqref{Schr_eq} with initial condition
\begin{equation}\label{thm_ini_trafo}
\Psi_N(0) = W\big(\sqrt{N}\alpha(0)\big) \sum_{k=0}^N \varphi(0)^{\otimes (N-k)} \otimes_s \chi^{(k)}_0(0)
\in \Hilbert_N,
\end{equation}
where $(\varphi(0),\alpha(0)) \in H^2(\mathbb R^3) \times L^2_1(\mathbb R^3)$ and $\chi_0(0) \in \mathcal G$ satisfy $ \norm{ \varphi(0) } = 1 $ and $\norm{\chi_0}=1$. We assume that there is a constant $C>0$ such that
\begin{equation}\label{thm_moments_assumption}
\big\| ( \Number + 1 )^{3/2} \chi_0(0) \big\| \leq C.
\end{equation}
Then there are constants $C_1,C_2 > 0$ such that
\begin{equation}\label{main_result}
\norm{\Psi_N(t) - \Psi_{N}^{(0)}(t)} \leq C_1 e^{C_2 t} N^{-1/2},
\end{equation}
where
\begin{equation}\label{def Psi_N^(0)(t)}
\Psi_N^{(0)}(t) = W\big(\sqrt{N}\alpha(t)\big) \sum_{k=0}^N \varphi(t)^{\otimes (N-k)} \otimes_s \chi^{(k)}_0(t),
\end{equation}
with $\varphi(t)$, $\alpha(t)$, and $\chi_0(t)$ solutions to Equations~\eqref{varphi_eq}, \eqref{alpha_eq}, and \eqref{Bog_Schr}, respectively.
\end{theorem}

\begin{remark} By density of $D(\mathcal N^{3/2} ) \cap \Gock$ in $\mathcal G$, the above statement extends to any $N$-independent state $\chi_0(0) \in \mathcal G$ if we omit the explicit rate of convergence in \eqref{main_result}. More precisely, for $\Psi_N(0)$ as in \eqref{thm_ini_trafo} with $\chi_0(0)$ normalized but not necessarily satisfying \eqref{thm_moments_assumption}, it follows that
\begin{align}
\lim_{N\to \infty} \norm{\Psi_N(t) - \Psi_{N}^{(0)}(t)}  = 0.
\end{align}
\end{remark}

\begin{remark} Since we choose $\chi_0(0) = (\chi_0^{(k)}(0))_{k\ge 0}$ normalized to one, this does not necessarily hold for $\Psi_N(0)$. However, it is easy to verify that $\norm{\Psi_N(0)} \to 1$ as $N\to \infty$. 
\end{remark}

\subsection{Higher-Order Corrections}

Formally, one can expand the Hamiltonian \eqref{full_H} by using the Taylor series $\sqrt{1-x} = \sum_{n=0}^{\infty} c_n x^n$ (see \eqref{def: coefficients cn} for the definition of $c_n$). This yields
\begin{align}
\begin{split}
H(t) &= \int dx\, b^{*}(x) h(t) b(x) + \int dk\, \omega(k) a^{*}(k) a(k) \\
&\quad + \sum_{n=0}^{\infty} N^{-n} c_n \int dx \int dk\,  
K(t,k,x) \big( a^{*}(k) + a(-k) \big)  b^*(x)   \Number_b^n + \hc  \\
&\quad + N^{-1/2} \int dx\, b^{*}(x) \Big( q(t) \widehat{\Phi}(x) q(t) - \scp{\varphi(t)}{\widehat{\Phi} \varphi(t)} \Big) b(x) \\
&= \sum_{\ell=0}^{\infty} N^{-\ell/2} H_\ell(t)
\end{split}
\end{align}
(see \eqref{def of H_1}-\eqref{def of H_2n+1}). The rigorous expansion with explicit remainder estimates is given in Lemma~\ref{Taylor_H_with_remainder}. In addition, we formally expand the wave function $\chi(t) \in \mathcal G$ in a power series in $N^{-1/2}$, i.e.,
\begin{equation}
\chi(t) = \sum_{\ell = 0}^{\infty} N^{-\ell/2} \chi_{\ell}(t), \quad \chi_\ell(t) \in \mathcal G.
\end{equation}
Later, we will only consider this power series truncated at some $r \in \NNN_0$, since we do not expect the series to converge. Then the Schr\"odinger equation
\begin{equation}
i \partial_t \chi(t) = H(t) \chi(t)
\end{equation}
leads in each order in $N^{-1/2}$ to the equations
\begin{align}\label{chil_equation}
i \partial_t \chi_{\ell}(t) = H_0(t)\chi_{\ell}(t) + \sum_{m=0}^{\ell-1} H_{\ell-m}(t) \chi_m(t).
\end{align}
Let us denote by $U_0(t,s)$ the unitary time evolution generated by the Bogoliubov Hamiltonian $H_0(t)$. Then Equation~\eqref{chil_equation} in integral form reads
\begin{align}
\chi_{\ell}(t) &= U_0(t,0)\chi_{\ell}(0) - i \sum_{m=0}^{\ell-1} \int_0^t \d s U_0(t,s) H_{\ell-m}(s) \chi_m(s) \nonumber \\
&= U_0 (t,0)\chi_{\ell}(0) - i \sum_{m=0}^{\ell-1} \int_0^t \d s \widetilde{H}_{\ell-m}(s,t) U_0(t,s) \chi_m(s),
\end{align}
where
\begin{equation}
\widetilde{H}_m (s,t) = U_0(t,s) H_m (s) U_0 (s,t).
\end{equation}
This iteration can be solved by a straightforward computation (see also \cite[Proposition~3.2]{bossmann}). We find
\begin{align}\label{chil_def}
\chi_{\ell}(t) &= U_ 0(t,0) \chi_{\ell}(0) + \sum_{m=0}^{\ell-1} \sum_{k=1}^{\ell-m} \sum_{\substack{\alpha \in \NNN^k \\ |\alpha| = \ell-m}} (-i)^k \int_{\Delta_k} \d s^{(k)} \prod_{i=1}^k \widetilde{H}_{\alpha_i}(s_i,t) U_0 (t,0) \chi_m(0),
\end{align}
where we abbreviated $\d s^{(k)} = \d s_1 \cdots \d s_k$, and $\Delta_k$ is the region $[0,t]^k$ with $s_{i+1} \leq s_i$ for all $i=1,\ldots,k-1$. Our main assumption on the initial data is the following.

\begin{assumption}\label{main_number_op_assumption}
Let $r\ge 1$ and assume that for each $\ell \in \{0,...,r\}$, the state $\chi_\ell(0) \in \mathcal G$ is in the domain of any power of the number operator $\mathcal N$ with uniform bounds as $N\to \infty$, that is, for all $\ell \in \{0,...,r\}$ and $n \in \NNN_0$ there are $C(\ell,n) > 0$ such that
\begin{align}\label{equation_assumption}
\norm{ (\mathcal N + 1 )^{n} \chi_\ell(0) } \leq C(\ell,n)
\end{align}
for all $N \geq 1$. Moreover, let $(\varphi(0),\alpha(0)) \in H^2(\mathbb R^3) \times L^2_1(\mathbb R^3)$.
\end{assumption}

Note that naturally one would consider $N$-independent coefficients $\chi_\ell(0)$, such that the uniformity in $N$ in \eqref{equation_assumption} is given; however, the following theorem remains true if the $\chi_\ell(0)$ depend on $N$ but satisfy \eqref{equation_assumption} uniformly in $N$. For mean-field bosons interacting via a class of two-body potentials, it was shown in \cite{BPS2021} that low-energy eigenstates for particles in a suitable trap can indeed be expanded into a series, with $N$-independent $\chi_\ell(0)$, and we expect that a similar result holds true for the Nelson model.

Our main result on the higher order corrections is the following.

\begin{theorem}\label{thm_higher_orders}
Let $\Psi_N(t)$ be the solution to the Schr\"odinger equation \eqref{Schr_eq} with initial condition $\Psi_N(0) \in \Hilbert_N$. Let $r \in \NNN_0$, and
\begin{equation}\label{thm_ini_state_higher_orders}
\Psi_N^{(r)}(0) = W\big(\sqrt{N}\alpha(0)\big) \sum_{k=0}^N \varphi(0)^{\otimes (N-k)} \otimes_s \bigg( \sum_{\ell = 0}^{r} N^{-\ell/2} \chi_{\ell}^{(k)}(0) \bigg),
\end{equation}
and let Assumption~\ref{main_number_op_assumption} hold. Then for all $\Psi_N(0)$ with
\begin{equation}
\norm{\Psi_N(0) - \Psi_N^{(r)}(0)} \leq C_r N^{-(r+1)/2}
\end{equation}
for some $C_r > 0$, there is a $\widetilde{C}_r > 0$ such that
\begin{equation}
\norm{\Psi_N(t) - \Psi_N^{(r)}(t)} \leq \widetilde{C}_r e^{\widetilde{C}_r t} N^{-(r+1)/2},
\end{equation}
where
\begin{equation}
\Psi_N^{(r)}(t) = W\big(\sqrt{N}\alpha(t)\big) \sum_{k=0}^N \varphi(t)^{\otimes (N-k)} \otimes_s \bigg( \sum_{\ell = 0}^{r} N^{-\ell/2} \chi_{\ell}^{(k)}(t) \bigg) ,
\end{equation}
with $(\varphi(t)$, $\alpha(t))$ satisfying \eqref{varphi_eq}-\eqref{alpha_eq} and $\chi_{\ell}(t)$ being defined as in Equation~\eqref{chil_def}.
\end{theorem}

To formulate our last result, we consider the unitary propagator $U(t,s)$ being defined by
\begin{equation}
i\partial_t U(t,s) = H(t) U(t,s) .
\end{equation} 
With the formal ansatz
\begin{equation}
U(t,s) = \sum_{\ell=0}^\infty N^{-\ell/2} U_{\ell}(t,s)
\end{equation}
we obtain similarly to \eqref{chil_def} that
\begin{align}\label{Uell_def}
U_\ell(t,s) = \sum_{k=1}^{\ell} \sum_{\substack{\alpha \in \NNN^k \\ |\alpha| = \ell}} (-i)^k \int_{\Delta_k(s)} \d s^{(k)} \prod_{i=1}^k \widetilde{H}^{(\alpha_i)}(s_i,t) U_0 (t,s),
\end{align}
where $\Delta_k(s)$ is the region $[s,t]^k$ with $s_{i+1} \leq s_i$ for all $i=1,\ldots,k-1$. When $\chi(0) = \sum_{\ell=0}^\infty N^{-\ell/2} \chi_\ell(0)$, we recover the expression \eqref{chil_def} via $\chi_\ell(t) = \sum_{m=0}^\ell U_{\ell-m}(t,0) \chi_m(0)$. Note that $U(t,s)$ and $U_0(t,s)$ are unitary two-parameter groups, but $U_\ell (t,s)$ for $\ell \geq 1$ is generally not unitary. Rather, $U(t,t) = \id = U_0(t,t)$ yields $U_\ell(t,t) = 0$, and the group property $U(t,s)U(s,r)$ yields $U_\ell(t,r) = \sum_{k=0}^{\ell} U_k(t,s) U_{\ell-k}(s,r)$.

\begin{theorem}\label{thm_higher_orders_unitary_group}
Let $U(t,s)$ be the unitary two-parameter group generated by $H(t)$, and let $\chi \in \Gock$ (possibly $N$-dependent) be such that
\begin{equation}\label{ass_number_unitary}
\forall n \in \NNN_0: \quad \sup_{N\ge 1} \scp{\chi}{(\Number+1)^{n} \chi} < \infty.
\end{equation}
Then for all $r \in \NNN_0$ and $t,s\in\RRR$,
\begin{equation}
\norm{\left( U(t,s) - \sum_{\ell=0}^r N^{-\ell/2} U_\ell(t,s)\right)\chi} \leq \widetilde{C}_r e^{\widetilde{C}_r (|t|+|s|)} N^{-(r+1)/2},
\end{equation}
for some $\widetilde{C}_r > 0$, with $U_\ell(t,s)$ as defined in Equation~\eqref{Uell_def}.
\end{theorem}

\section{Application of Main Results}
\label{sec:discussion-results}

In this section we outline two important applications of our main results. In the first part, we use that the Bogoliubov Hamiltonian is a quadratic operator which allows us to explicitly compute the action of the Bogoliubov time evolution on the creation and annihilation operators. In particular, this yields a simple method to compute correlation functions within the Bogoliubov approximation. The explicit form of the higher order terms from \eqref{chil_def} allows us to approximate correlation functions also to higher order. We use this in the second part, where we apply Theorem \ref{thm_higher_orders}  to obtain a set of coupled PDEs that describe the next-order correction to the time evolution of the one-particle reduced density matrices.

\subsection{Explicit action of the Bogoliubov time evolution}

To determine how the Bogoliubov time evolution acts on the  creation and annihilation operators, it is convenient to start by writing the Bogoliubov Hamiltonian as an operator that is quadratic in one type of (two-component) creation/annihilation operators. To this end, we introduce the Fock spaces
\begin{align}
\mathcal L= \bigoplus_{n=0}^{\infty}  \mathfrak{h}_2^{\otimes_s n} \quad \text{and} \quad \mathcal L_\perp 
= \bigoplus_{n=0}^{\infty}  \mathfrak{h}_{\perp \varphi(t)}^{\otimes_s n} \subset \mathcal L,
\end{align} 
with
\begin{align}
\mathfrak{h}_2 &= L^2(\mathbb{R}^3) \oplus L^2(\mathbb{R}^3) \quad \text{and} \quad 
\mathfrak{h}_{\perp \varphi(t)} = \{\varphi(t)\}^{\perp} \oplus L^2(\mathbb{R}^3)  \subset \mathfrak{h}_2.
\end{align}
We denote the creation and annihilation operators on these spaces by $Z^*$ and $Z$. Note that the space $\mathcal L_\perp$ is unitarily equivalent to $\mathcal G$, that is, there is a unitary (time-independent) map $V$ such that $V  \Omega_{\Gock}  = \Omega_{\mathcal L_\perp}$ and

\begin{subequations}
\begin{align}
Z(f \oplus g) &= V  \left( b(f) \otimes \id + \id \otimes a(g) \right) V ^*, \\
Z^{*}(f \oplus g) &= V  \left( b^{*}(f) \otimes \id + \id \otimes a^{*}(g) \right) V ^*
\end{align}
\end{subequations}
(see also \cite[Chapter 4.2]{MNO2019}). As usual, we define the second quantization of a linear operator $A$ on $\mathfrak h_2$ by
\begin{align}
d\Gamma (A)  = \sum_{n,m=0}^\infty \scp{u_n}{A u_m} Z^* (u_n) Z(u_m)
\end{align}
where $(u_n)_{n\in \mathbb N}\subset \mathfrak{h}_2$ is a suitable ONB. The number operator, for instance, is given by $\mathcal N = d\Gamma(1)$. The action of the Bogoliubov Hamiltonian on $\mathcal L_\perp$ can then be written as
\begin{align}\label{H_Bog_2X}
\mathbb H_0(t) & = V H_0(t) V^* \nonumber \\
&  = d\Gamma(A(t)) +  \frac{1}{2} \Bigg( \sum_{n,m} \scp{u_n}{B(t) \overline{ u_m} } Z^* (u_n) Z^* (u_m) +  \sum_{n,m} \overline{ \scp{u_n}{B(t) \overline{ u_m } }} Z (u_n) Z (u_m) \Bigg)
\end{align}
where
\begin{equation}
\label{eq: def A(t)}
A(t)= A_1 + A_{2}(t)= \left(\begin{array}{cc} h(t) & \overline{K^{*}_{-}(t)} \\ \overline{K_{-}(t)} & \omega \end{array}\right)=\left(\begin{array}{cc} -\Delta & 0 \\ 0 & \omega \end{array}\right) + \left(\begin{array}{cc} \Phi(t,x) - \mu(t) & \overline{K^{*}_{-}(t)} \\ \overline{K_{-}(t)} &  0 \end{array}\right)
\end{equation}
and 
\begin{equation}
\label{eq: def B(t)}
B(t) = \left(\begin{array}{cc} 0 & \overline{K^{*}(t)} \\ K(t) & 0 \end{array}\right)
\end{equation}
are defined as operators on $\mathfrak{h}_2$, $(u_n)_{n\in \NNN } \subset D(A_1)$ is an ONB of $\mathfrak{h}_2$ and $K_{-}(t)$ is the operator on $L^2(\mathbb{R}^3)$ with kernel $K_{-}(t,k,x) = K(t,-k,x)$. We can now apply well-known results about the time evolution of quadratic Hamiltonians to our model \eqref{H_Bog_2X}; we refer to \cite{JPS2007} and \cite{bossmann} for an introduction to the topic and an overview of some of the results.
For
\begin{align}
F = f \oplus J g = f \oplus \overline{g} = f_1 \oplus f_2 \oplus \overline{g_1} \oplus \overline{g_2} \in \mathfrak{h}_2 \oplus \mathfrak{h}_2 ,
\end{align}
let the generalized creation and annihilation operators $A^{*}(F)$ and $A(F)$ be defined by
\begin{align}
A^{*}(F) = Z^{*}(f) + Z(g) = A(\mathcal{J} F), \quad A(F) &= Z(f) + Z^{*}(g),
\label{eq: definition generalized creation operatorX}
\end{align}
where the operator
\begin{align}
\mathcal{J} =  \left( \begin{array}{rr}
0 & J \\
J & 0 \\
\end{array}\right)
\end{align}
acts on $\mathfrak{h}_2 \oplus \mathfrak{h}_2$ and $J: \mathfrak{h}_2 \rightarrow \mathfrak{h}_2$, $J(g_1(k), g_2(x)) = (\overline{g_1(k)}, \overline{g_2(x)})$ is the complex conjugation map. We further define $\mathcal S$ on $\mathfrak{h}_2 \oplus \mathfrak{h}_2 $ by
\begin{align}
\mathcal{S} =  \left( \begin{array}{rr}
1 & 0   \\
0 & -1 \\
\end{array}\right) .
\end{align}
The generalized creation and annihilation operators satisfy the commutation relations
\begin{align}
\label{eq: commutation relation generalized creation and annihilation operators}
\left[ A(F) , A^{*}(G) \right] = \scp{F}{\mathcal{S} G}_{\mathfrak{h}_2 \oplus \mathfrak{h}_2} , \qquad 
\left[ A(F) , A(G) \right] =  \scp{F}{\mathcal{S} \mathcal{J} G}_{\mathfrak{h}_2 \oplus \mathfrak{h}_2} .
\end{align}
Next we recall the notion of Bogoliubov maps and Bogoliubov transformations. A bounded operator
\begin{align}
\mathcal{V}: \mathfrak{h}_2 \oplus \mathfrak{h}_2 \rightarrow \mathfrak{h}_2 \oplus \mathfrak{h}_2 
\end{align}
is called a Bogoliubov map if it satisfies
\begin{align} \label{Conditions Bog map}
\mathcal{V}^* \mathcal{S} \mathcal{V} = \mathcal{S} = \mathcal{V} S \mathcal{V}^* \quad \text{and} \quad
\mathcal{J} \mathcal{V} \mathcal{J} = \mathcal{V} .
\end{align}
Equivalently, $\mathcal{V}$ is a Bogoliubov map if it has the block form
\begin{align}
\mathcal{V} =  \left( \begin{array}{rr}
	u & \overline{v} \\
	v & \overline{u} \\
	\end{array}\right) 
	\quad \text{with} \quad  u,v: \mathfrak{h}_2  \rightarrow \mathfrak{h}_2, \quad \overline{v} = J v J , \quad  \overline{u} = J u J,
\end{align}
and $u,v$ satisfy the conditions $u^* u = \id + v^* v, \, \, u u^* = \id + \overline{v} \, \overline{v}^*, \, \,
	v^* \overline{u} = v^* \overline{v} , \, u v^* = \overline{v} \, \overline{u}^* .$

A well-known result about Bogoliubov maps states that there exists a unitary transformation $\mathbb{U}_{\mathcal{V}}: \mathcal L \rightarrow \mathcal L$ satisfying 
\begin{align}
\label{eq: Bogoliubov tranformation action on AX}
\mathbb{U}_{\mathcal{V}}^* A(F) \mathbb{U}_{\mathcal{V}} = A(\mathcal{V} F)\quad \forall F \in \mathfrak{h}_2 \oplus \mathfrak{h}_2,
\end{align}
if and only if it satisfies the Shale--Stinespring condition
\begin{align}
\norm{v}_{{\rm HS}}^2 = \text{Tr}_{\mathfrak h_2} \left( v^* v \right) < \infty.
\end{align}
In this case, $\mathcal{V}$ is called unitarily implementable and its implementation $\mathbb U_{\mathcal V}$ will be called the Bogoliubov transformation associated to $\mathcal V$.

It is a fundamental property of the dynamics generated by (suitable) quadratic Hamiltonians on Fock space that they are equivalent to
time-dependent Bogoliubov transformations. 
More precisely, in our case, the
unitary propagator $U_0(t,t_0) : \mathcal{G} \rightarrow \mathcal{G}$ that corresponds to the Schr\"odinger
equation \eqref{Bog_Schr},\footnote{Existence and further properties of
  $U_0(t,t_0)$ will be discussed in Section \ref{sec:well-posedn-bogol}.}
is unitarily equivalent to the Bogoliubov transformation $U_{\mathcal{V}(t,t_0)} : \mathcal{L}_{\perp} \rightarrow \mathcal{L}_{\perp}$, i.e., $U_0(t,t_0) = V^* U_{\mathcal{V}(t,t_0)} V$, where the time-dependent Bogoliubov map $\mathcal V(t,t_0)$ solves the PDE
\begin{align}
\begin{cases}\label{eq: PDE Bogoliubov map}
\quad i \partial_t \mathcal{V}(t,t_0) & = \ \ \mathcal{A}(t) \mathcal{V}(t,t_0) \\
\quad \quad \mathcal{V}(t_0,t_0) &= \ \ \id ,
\end{cases}
\end{align}
with
\begin{align}
\mathcal{A}(t) = \left( \begin{matrix}
A(t)  & -    B(t) \\[1.5mm]
\overline{B(t)} & - \overline{A(t)}  
\end{matrix}\right),
\end{align}
and $A(t)$ and $B(t)$ defined by \eqref{eq: def A(t)} and \eqref{eq: def B(t)}, respectively. Note that Equation \eqref{eq: PDE Bogoliubov map} reduces the Bogoliubov equation to a PDE on $\mathfrak{h}_2 \oplus \mathfrak{h}_2$
and that 
\begin{align}
\label{eq: Bogoliubov tranformation action on AX2}
V U_{0}(t,0)^* V^{*} A( F) V U_{0}(t,0) V^{*} =  A(\mathcal{V} (t,0) F)\quad \forall F \in \mathfrak{h}_2 \oplus \mathfrak{h}_2 .
\end{align}
Equation \eqref{eq: Bogoliubov tranformation action on AX2} is our main result of this section. It is very convenient since it allows one to explicitly compute expectation values of a large class of observables in the state $\chi_0(t)$ or, equivalently, in the state $\Psi_N^{(0)}(t)$ defined by \eqref{def Psi_N^(0)(t)}.

\subsection{Equations for the reduced density and next-order corrections}

The expansion of the wave function in Theorem~\ref{thm_higher_orders} implies a corresponding expansion of correlation functions, as has been discussed in \cite{bossmann} for pair-interacting bosons. Here, we provide the next-order correction to the reduced one-particle density matrices, as one important application of Theorem~\ref{thm_higher_orders}. The one-particle reduced density matrix of the condensate particles is defined as
\begin{equation}
\gammapart = \tr_{2\ldots,N} \ketbr{\Psi_N(t)},
\end{equation}
and of the field as
\begin{equation}
\gammafield(k,k') = N^{-1}\scp{\Psi_N(t)}{a^*(k') a(k)\Psi_N(t)}_{\Hilbert_N}.
\end{equation}
These reduced densities can be used to compute expectation values of bounded one-body operators. We now express the reduced densities in terms of the excitation vector $\chi_{\leq N}(t)$ by using the unitary $U_N(t) = \widetilde{U}_N(t) \otimes W^*(\sqrt{N}\alpha(t))$ that was explicitly defined in \eqref{psi_to_excitations} (see Appendix \ref{section: more details on the excitation Fock space} for more details). A straightforward computation (using Lemma \ref{lemma:properties U}) shows that
\begin{align}\label{gamma_part_exc}
\gammapart(x,x') &= N^{-1} \scp{\chi_{\leq N}(t)}{\widetilde{U}_N(t) b^*(x') b(x) \widetilde{U}^*_N(t) \chi_{\leq N}(t)}_{\Hilbert_N} \nonumber\\
\begin{split}&= \varphi(t,x) \overline{\varphi(t,x')} + N^{-1/2} \left( \varphi(t,x) \overline{\betapart(x')} + \betapart(x) \overline{\varphi(t,x')} \right) \\
&\qquad + N^{-1} \left( \gammapartFock(x,x') - \varphi(t,x) \overline{\varphi(t,x')} \, \tr\, \gammapartFock  \right),\end{split}
\end{align}
where, for any $\Phi \in \mathcal{G}$,
\begin{equation}
\betapartPhi(x) := \scp{\Phi}{\sqrt{[1-N^{-1}\Number_b]_+}\, b(x) \Phi}_{\mathcal{G}}, \quad \gammapartFockPhi(x,x') := \scp{\Phi}{b^*(x') b(x) \Phi}_{\mathcal{G}},
\end{equation}
and
\begin{align}\label{gamma_field_exc}
\gammafield(k,k') &= N^{-1} \scp{\chi_{\leq N}(t)}{W^*(\sqrt{N}\alpha(t)) a^*(k') a(k)W(\sqrt{N}\alpha(t)) \chi_{\leq N}(t)}_{\Hilbert_N} \nonumber\\
\begin{split}&= \alpha(t,k) \overline{\alpha(t,k')} + N^{-1/2} \left( \alpha(t,k) \overline{\betafield(t,k')} + \betafield(t,k) \overline{\alpha(t,k')} \right) \\
&\qquad + N^{-1} \gammafieldFock(k,k'),\end{split}
\end{align}
where, for any $\Phi \in \mathcal{G}$,
\begin{equation}
\betafieldPhi(k) := \scp{\Phi}{a(k)\Phi}_{\mathcal{G}}, \quad \gammafieldFockPhi(k,k') := \scp{\Phi}{a^*(k') a(k)\Phi}_{\mathcal{G}}.
\end{equation}
Equations~\eqref{gamma_part_exc} and \eqref{gamma_field_exc} tell us how the reduced densities of the particles and the quantum field are related to reduced densities and one-point functions of the excitations.

The leading orders of the reduced density matrices of the excitations, $\gammapartFockz$ and $\gammafieldFockz$, can be computed via Bogoliubov theory. It is easier to directly compute the time evolution of the generalized one-particle density matrix $\Gamma_{\chi_{0}(t)}$, defined via
\begin{equation}
\scp{F_1}{\Gamma_{\chi_{0}(t)}F_2}_{\mathfrak{h}_{2} \oplus \mathfrak{h}_{2}} := \scp{\chi_0(t)}{V^{*} A^*(F_2) A(F_1) V \chi_0(t)}_{\mathcal{G}}
\end{equation}
for any $F_1,F_2 \in \mathfrak{h}_{\perp \varphi(t)} \oplus \mathfrak{h}_{\perp \varphi(t)}$, where $A^*, A$ are the generalized creation and annihilation operators defined in \eqref{eq: definition generalized creation operatorX}. Using the fact that $\chi_0(t) = V^{*} \mathbb U_{\mathcal{V}(t,0)} V \chi_0(0)$ with the Bogoliubov map $\mathcal{V}(t, 0)$ that solves \eqref{eq: PDE Bogoliubov map}, a direct computation yields
\begin{align}
\Gamma_{\chi_{0}(t)} = \mathcal{S} \mathcal{V}(t, 0) \mathcal{S} \Gamma_{\chi_{0}(0)} \mathcal{S} \mathcal{V}^*(t, 0) \mathcal{S}, 
\end{align}
or
\begin{align}\label{eq_Gamma}
i\partial_t \Gamma_{\chi_{0}(t)} =  \mathcal{A}(t)^* \Gamma_{\chi_{0}(t)} - \Gamma_{\chi_{0}(t)} \mathcal{A}(t).
\end{align}
From  $\Gamma_{\chi_{0}(t)}$ we can read off all possible two-point functions of $\chi_0(t)$, in particular $\gammapartFockz$ and $\gammafieldFockz$.

Let us next consider the one-point functions $\betapart$ and $\betafield$. The leading order of $\betapart$ is $\betapartzz(t,x) := \scp{\chi_{0}(t)}{b(x) \chi_{0}(t)}_{\mathcal{G}}$. If we assume that $\chi_0(0) = V^{*} \mathbb U_{\mathcal{V}} \Omega$ for some Bogoliubov map $\mathcal{V}$ (as defined in the previous section), i.e., that $\chi_0(0)$ is quasi-free, then $\chi_0(t) = V^{*} \mathbb U_{\mathcal{V}(t,0)} \mathbb U_{\mathcal{V}} \Omega$ is also quasi-free, as the composition of two Bogoliubov transformations defines again a Bogoliubov transformation. Therefore, $\betapartzz = 0$ due to Wick's rule \cite{JPS2007}.

\begin{remark}\label{remark_ini_data}
We could more generally assume that $\chi_0(0) = V^{*} U_{\mathcal{V}} Z^*(f_1) \cdots Z^*(f_n) \Omega$ for some orthonormal $f_1,\ldots,f_n \in \mathfrak{h}_{\perp \varphi(0)}$, i.e., that $\chi_0(0)$ is a state with a fixed number $n$ of excitations. Then $\betapartzz$ is a $(2n+1)$-point function of a quasi-free state, i.e., still $\betapartzz = 0$ due to Wick's rule. Such states are the prediction of Bogoliubov theory for the low-energy excited states of trapped systems; see \cite{BPS2021} for references in the case of bosons with two-body interaction.

The structure of trapped initial data $\chi_\ell(0)$ for $\ell \geq 1$ follows from time-independent perturbation theory, which was proven in \cite{BPS2021} for a class of two-body interactions. We expect similar results to hold for our model, in particular that
\begin{equation}\label{chi_l_zero}
\chi_{\ell}(0) = V^{*} \sum_{\substack{m=0 \\ m+\ell\,\text{even}}}^{3\ell} \sum_{\mu = 0}^m G_{m,\mu}^{(\ell)} V \chi_0(0),
\end{equation}
where $G_{m,\mu}^{(\ell)}$ is the quantization of a bounded operator $g_{m,\mu}^{(\ell)}:(\mathfrak{h}_{\perp \varphi(0)})^{m-\mu} \to (\mathfrak{h}_{\perp \varphi(0)})^{\mu}$, i.e.,
\begin{equation}
G_{m,\mu}^{(\ell)} = \int \d x^{(\mu)} \int \d y^{(m-\mu)} g_{m,\mu}^{(\ell)}(x^{(\mu)}; y^{(m-\mu)}) Z^*(x_1) \ldots Z^*(x_\mu) Z(y_1) \ldots Z(y_{m-\mu}),
\end{equation}
with $x^{(\mu)} := (x_1,\ldots,x_\mu)$. So for even $\ell+n$ ($n$ being the number of excitations in $\chi_0(0)$) the state $\chi_\ell(0)$ has only an even number of excitations, whereas for odd $n+\ell$ it has an odd number. Note that the time evolution \eqref{chil_def} indeed preserves the structure of Equation~\eqref{chi_l_zero}, i.e., also
\begin{equation}\label{chi_l_t}
\chi_{\ell}(t) = V^{*} \sum_{\substack{m=0 \\ m+\ell\,\text{even}}}^{3\ell} \sum_{\mu = 0}^m G_{m,\mu}^{(\ell)}(t) V \chi_0(t)
\end{equation}
for some $G_{m,\mu}^{(\ell)}(t)$. Assuming \eqref{chi_l_zero}, all correlation functions of the type $\scp{\chi_\ell(t)}{Z_1^{\sharp} \ldots Z_k^{\sharp} \chi_m(t)}$, where $Z^{\sharp} \in \{ Z, Z^* \}$, vanish for $\ell+m+k$ odd due to Wick's rule.
\end{remark}

With the assumption from Remark~\ref{remark_ini_data}, all terms of order one and of order $N^{-1}$ vanish in the expansion of $\betapart$, so
\begin{align}
\betapart = N^{-1/2} \betapartzo(t) + O(N^{-3/2}), ~~\text{with}~~ \betapartzo(t,x) = \scp{\chi_{0}(t)}{b(x) \chi_{1}(t)}_{\mathcal{G}} + \scp{\chi_{1}(t)}{b(x) \chi_{0}(t)}_{\mathcal{G}}.
\end{align}
In the same way,
\begin{align}
\betafield = N^{-1/2} \betafieldzo(t) + O(N^{-3/2}), ~~\text{with}~~ \betafieldzo(t,k) = \scp{\chi_{0}(t)}{a(k)\chi_{1}(t)}_{\mathcal{G}} + \scp{\chi_{1}(t)}{a(k)\chi_{0}(t)}_{\mathcal{G}}.
\end{align}
A direct computation shows that $\betapartzo$ and $\betafieldzo$ solve the coupled PDEs
\begin{subequations}
\begin{align}\label{eq_beta_part}
\begin{split}i \partial_t \betapartzo(t,x) &= h(t) \betapartzo(t,x) + \int \d k \left( K(t,k,x)\overline{\betafieldzo(t,k)} + K(t,-k,x) \betafieldzo(t,k)\right)  \\
&\quad + \int \d k \d y F(t,x,k,y) \scp{\chi_0(t)}{\big(a(k)^* b(y) + a(-k) b(y)\big) \chi_0(t)}_{\mathcal{G}}\end{split} \\
\label{eq_beta_field}\begin{split}i \partial_t \betafieldzo(t,k) &= \omega(k) \betafieldzo(t,k) + \int \d x \left( K(t,k,x) \overline{\betapartzo(t,x)} + \betapartzo(t,x) \overline{K(t,-k,x)} \right) \\
&\quad + \int \d x \d y F(t,x,k,y) \scp{\chi_0(t)}{b^*(x) b(y) \chi_0(t)}_{\mathcal{G}},\end{split} 
\end{align}
\end{subequations}
where
\begin{equation}
F(t,x,k,y) := \eta(k) \int \d z \left( q(t,x,z) e^{-2\pi ikz} q(t,z,y) - \delta(x-y) e^{-2\pi ikz} |\varphi(t,z)|^2 \right).
\end{equation}
The two-point functions such as $\scp{\chi_0(t)}{a(k)^* b(y)\chi_0(t)}_{\mathcal{G}}$ can be read off from the corresponding entry of $\Gamma_{\chi_{0}(t)}$.

To summarize, for initial data as in Remark~\ref{remark_ini_data} we have found that
\begin{align}
\gammapart &= \ketbr{\varphi(t)} + N^{-1} \left( \ketbra{\varphi(t)}{\betapartzo(t)} + \ketbra{\betapartzo(t)}{\varphi(t)} + \gammapartFockz - \ketbr{\varphi(t)} \, \tr \gammapartFockz \right) \nonumber\\
&\quad + O(N^{-2}),
\end{align}
and
\begin{align}
\gammafield &= \ketbr{\alpha(t)} + N^{-1} \left( \ketbra{\alpha(t)}{\betafieldzo(t)} + \ketbra{\betafieldzo(t)}{\alpha(t)} + \gammafieldFockz \right) + O(N^{-2}),
\end{align}
where the one-point functions $\betapartzo(t)$ and $\betafieldzo(t)$ solve Equations~\eqref{eq_beta_part} and \eqref{eq_beta_field}, and the two-point functions $\gammapartFockz$ and $\gammafieldFockz$ can be read off from $\Gamma_{\chi_{0}(t)}$, the solution of Equation~\eqref{eq_Gamma}.

\section{Well-posedness of the Bogoliubov Equation\label{sec_Bog_eq}}
\label{sec:well-posedn-bogol}

Let us now discuss the well-posedness of the Cauchy problem
\begin{align}
\label{Bogoliubov_SE 2}
\begin{cases}
 \,  i\partial_t \chi(t) = \mathbb H_0(t) \chi (t)\\[1.5mm]
\quad   \chi(t_0)= \chi_0 \in \mathcal L.
\end{cases}\;
\end{align}
Since $H_0(t)$ is unitarily equivalent to $\mathbb H_0(t)$ (see \eqref{H_Bog_2X}), this provides the well-posedness of the Bogoliubov equation \eqref{Bog_Schr}. The main ingredient in the proof is a regularity estimate for the time-dependent part of $H_0(t)$; more precisely, below we shall show that $A_2(t)$ and $B(t)$, defined in \eqref{eq: def A(t)} and \eqref{eq: def B(t)}, satisfy
\begin{align}\label{eq: conditions in Bogoliubov dynamicsx}
\sup_{t\in K} \bigg( \norm{A_2(t)} + \norm{\frac{d}{dt}A_2(t)} + \norm{ B(t)}_{\rm HS} + \norm{\frac{d}{dt}B(t)}_{\rm HS} \bigg) < \infty 
\end{align} 
for any compact set $K\subset \mathbb R$, where $\norm{\cdot}_{\rm HS}$ denotes
the Hilbert--Schmidt norm. Given the above bound, the following theorem holds.

\begin{theorem}
Let $(\varphi(t), \alpha(t))$ be the unique solution of \eqref{varphi_eq}-\eqref{alpha_eq} with initial datum 
\begin{align}
(\varphi(t_0), \alpha(t_0)) \in H^2(\mathbb{R}^3) \times L^2_1(\mathbb{R}^3).
\end{align}
Then there exists a two-parameter unitary group $(U_0(t,t_0))_{t,t_0\in \mathbb{R}}$ on
$\mathcal L  $ such that for any $t,t_0\in \mathbb{R}$,
\begin{align}
\chi_{t_0}(t) = U_0(t,t_0 ) \chi
\end{align}
is the unique solution to \eqref{Bogoliubov_SE 2} in the following sense. For
any $\xi,\chi \in Q(d\Gamma( 1 + A_1 ) ) \subseteq \mathcal L $ (where $Q(\cdot)$ denotes the quadratic form domain), we have that $U_0 (t,t_0)\chi, U_0(t_0,t)\xi\in
Q(d\Gamma( 1 + A_1 ) )$, and therefore the function $\langle \xi , U_0 (t,t_0)\chi
\rangle_{\mathcal{L}}$ is differentiable with respect to both $t$ and $t_0$, with
\begin{subequations}
\begin{align}
  \label{eq:1x}
  i\partial_t \langle \xi  , U_0 (t,t_0)\chi \rangle_{\mathcal{L}} &= \langle \xi  , \mathbb H_0 (t) U_0  (t,t_0)\chi \rangle_{\mathcal{L}}\\
 \label{eq:2x}
  i\partial_{t_0} \langle \xi  , U_0  (t,t_0)\chi \rangle_{\mathcal{L}} &= - \langle \xi  , U_0 (t,t_0) \mathbb H_0 (t_0)\chi \rangle_{\mathcal{L}}\; .
\end{align}
\end{subequations}
Moreover, $U_0(t,t_0)$ satisfies the following properties.
\begin{itemize}
\item[(i)] For any $r\in \mathbb{R}$ there exists a constant $C_r>0$, such that
\begin{align}
  \label{eq:4x}
\lVert (d\Gamma(1)+1)^{r} U_0 (t,t_0) (d\Gamma(1)+1)^{-r}  \rVert_{}^{}\leq e^{C_r\lvert t-t_0  \rvert_{}^{}}\; 
\end{align}
for all $t,t_0 \in \mathbb R$.
\item[(ii)] For $\chi \in \mathcal L_\perp(t_0) $, we have $\chi_{t_0}(t) \in \mathcal L_\perp(t)$.
\end{itemize}
\end{theorem}

\begin{proof}
The existence of the two-parameter unitary group satisfying \eqref{eq:1x} and
\eqref{eq:2x} follows from a general result about the evolution generated by
a time-dependent Hamiltonian \cite[Theorem 8]{Lewin:2015a}. That
\cite[Theorem 8]{Lewin:2015a} can be applied in our setting under condition \eqref{eq: conditions in Bogoliubov dynamicsx} is shown in analogy to the first part of the proof of \cite[Proposition 7]{nam:2016}. For our purpose, the bound \eqref{eq:4x} is crucial. It can be proved by a Gronwall argument, following the same strategy used in the proof of \cite[Proposition 4.2]{falconi}. (For a similar bound, see also \cite[Theorem 8]{Lewin:2015a}.) In the following, we first outline the argument of the proof of property $(i)$, then of property $(ii)$, and finally we prove the bound \eqref{eq: conditions in Bogoliubov dynamicsx}.

We consider $\chi_{t_0}(t)$, with $\chi\in Q(d\Gamma( 1 + A_1 ) )$. Its derivative is an element of the Hilbert
space obtained from completing $\mathcal{L}$ w.r.t.\ the scalar product
\begin{equation}
  \langle \,\cdot\,  , \bigl(d\Gamma(1+A_1)\bigr)^{-1}\,\cdot\, \rangle_{\mathcal{L}}\;.
\end{equation}
Hence, $M(t,t_0):= \langle \chi_{t_0}(t)  , \chi_{t_0}(t)
\rangle_{\mathcal{L}}$ is differentiable w.r.t.\ both $t$ and $t_0$. Analogously,
for any $r\geq1$,
\begin{equation}
  M_r(t,t_0) := \langle \chi_{t_0}(t)  , \bigl(d\Gamma(1)+1\bigr)^{-r} \chi_{t_0}(t)\rangle_{\mathcal{L}}
\end{equation}
is differentiable w.r.t.\ both $t$ and $t_0$ as well. We have that
\begin{equation}
  i\partial_t M_r(t,t_0) = \langle \chi_{t_0}(t)  , [\bigl(d\Gamma(1)+1\bigr)^{-r},\mathbb{H}_0(t)] \chi_{t_0}(t)\rangle_{\mathcal{L}}\; .
\end{equation}
The commutation yields
\begin{equation}
  \lvert \partial_t M_r(t,t_0)  \rvert_{}^{} \leq 2C \langle \chi_{t_0}(t)  ,  \bigl(\bigl(d\Gamma(1)+1\bigr)^{-r}- \bigl(\mathrm{d}\Gamma(1)+3\bigr)^{-r}\bigr)\, d\Gamma(1)  \chi_{t_0}(t)\rangle_{\mathcal{L}}  \;,
\end{equation}
where $C=\sup_{\tau\in [t_0,t]}\norm{ B(\tau)}_{ \rm HS }$. Now, by spectral
calculus the inequality
\begin{equation}
  \bigl((n+1)^{-r}-(n+3)^{-r}\bigr)n \leq 2r(n+1)^{-r}\;,
\end{equation}
valid for all $n\in \mathbb{N}$, leads to
\begin{equation}
  \lvert \partial_t M_r(t,t_0)  \rvert_{}^{} \leq 4C\,r M_r(t,t_0)  \;.
\end{equation}
Gronwall's lemma gives
\begin{equation}
  \lVert (d\Gamma(1)+1)^{-r/2}\chi_{t_0}(t)  \rVert_{\mathcal{L}}^2\leq e^{4C\,r \lvert t-t_0  \rvert_{}^{}}\lVert (d\Gamma(1)+1)^{-r/2} \chi  \rVert_{\mathcal{L}}^2\;.
\end{equation}
The result is extended to $0\leq r\leq 1$ by interpolation, the case $r=0$ being
trivial. In addition, the bound is extended to any $\chi\in
\mathcal{L}$ by a density argument. Therefore, we can choose $\chi=(d\Gamma(1)+1)^{r/2}\xi$, for some $\xi\in
Q(d\Gamma(1)^r)$. In this case, the result reads
\begin{equation}
  \lVert (d\Gamma(1)+1)^{-r/2}U_0(t,t_0)  (d\Gamma(1)+1)^{r/2}\xi\rVert_{\mathcal{L}}^2\leq e^{4C\,r \lvert t-t_0  \rvert_{}^{}}\lVert \xi \rVert_{\mathcal{L}}^2\;,
\end{equation}
for any $t,t_0\in \mathbb{R}$. Again by density the result is extended to any $\xi\in
\mathcal{L}$. Hence,
\begin{equation}
  \lVert (d\Gamma(1)+1)^{-r/2} U_0 (t,t_0) (d\Gamma(1)+1)^{r/2}  \rVert_{}^{}\leq e^{2C\, r\lvert t-t_0  \rvert_{}^{}}\;.
\end{equation}
Since $U_0 (t,t_0)^{*}= U_0 (t_0,t)$, the bound \eqref{eq:4x} follows, setting
$C_r = 4C \lvert r \rvert_{}^{} $.

Next, we prove property $(ii)$. One computes
\begin{align}
\frac{d}{dt} \norm{ Z( \varphi(t) \oplus 0 ) \chi_{t_0}(t) }^2 = 2 \Im  \scp { Z( \varphi(t) \oplus 0 ) \chi_{t_0}(t)  }{ \big(  Z ( (- h(t) \varphi(t)) \oplus 0 )   +   Z( \varphi(t) \oplus 0 ) \mathbb H_0 (t) \big) \chi_{t_0}(t) } .
\end{align}
Using
\begin{align}
 Z( \varphi(t) \oplus 0) \mathbb H_0 (t)  = \mathbb H_0(t)  Z( \varphi(t) \oplus 0 ) + Z ( h(t) \varphi(t) \oplus 0 ) ,
\end{align}
one finds
\begin{align}
\frac{d}{dt} \norm{ Z( \varphi(t) \oplus 0 ) \chi_{t_0}(t) }^2 =  2 \Im  \scp { Z( \varphi(t) \oplus 0 ) \chi_{t_0}(t)  }{  \mathbb H_0 (t) Z( \varphi(t) \oplus 0 ) \chi_{t_0}(t) } = 0,
\end{align}
and thus $\norm{ Z( \varphi(t) \oplus 0 ) \chi_{t_0}(t) } = 0$ for all $\chi_{0}$ satisfying $\norm{ Z( \varphi(t_0) \oplus 0 ) \chi_{0} }= 0$. This implies $(ii)$.

It remains to show \eqref{eq: conditions in Bogoliubov dynamicsx}. Since
\begin{align}
\vert \Phi(t,x)\vert  \le 2  \norm{\eta}_{L^2}\norm{\alpha(t)}_{L^2},\quad 
\vert \mu(t) \vert \le  \norm{\eta}_{L^2}\norm{\alpha(t)}_{L^2},\quad 
\norm{ K(t,\cdot ,\cdot )}_{L^2(\mathbb R^6)} & \le  2 \norm{\eta}_{L^2},
\end{align}
we have $ \sup_{t \in K} ( \norm{A_2(t)} + \norm{B (t) }_{ \mathfrak{S}^2 } ) <\infty$ for any compact interval $K\subset \mathbb R$.

Using the Schr\"odinger--Klein--Gordon equations, we further find
\begin{align}
 i \partial_t \Phi(x,t)   &   =  \int dk\,  \eta(k) \omega(k) e^{2\pi ikx} \big( \alpha(t,k) - \overline{\alpha(t,- k)} \big),
\end{align} 
and hence, $\sup_{x\in \mathbb R^3} | \partial_t \Phi(x,t)  |  \le  2 \norm{\eta}_{L^2} \, \norm{\omega\alpha(t)}_{L^2}$, where $\norm{\omega\alpha(t)}_{L^2}$ is bounded by Lemma~\ref{lemma: solution theory}. With $\norm{\varphi(t)}_{L^2} = 1 $, we further get
\begin{align}
| \partial_t \mu(t) | \le  \sup_{x\in \mathbb R^3} \Big( \frac{1}{2} \vert \partial_t \Phi(t,x)\vert  + \vert \Phi(t,x)\vert \norm{h(t) \varphi(t)}_{L^2} \Big) . 
\end{align}
The last norm is bounded by estimating
\begin{align}\label{bound for hvarphi}
\sup_{t \in K} \norm{h(t ) \varphi(t) } \le \sup_{t \in K} \big( \norm{\Delta \varphi(t) } + (\vert \Phi(t,x) \vert +  \vert \mu(t) \vert ) \norm{\varphi(t)} \big) <\infty .
\end{align}
To prove the bound for $\partial_t K(t)$, we recall $K(t) = \widetilde K(t) (1-\overline{p(t)})$ and use
\begin{align}
\partial_t K(t) = ( \partial_t \widetilde{K}(t))  \overline{q(t)} - \widetilde{K}(t) \partial_t \overline{p(t)} . 
\end{align}
Then we proceed with $\sup_{t \in K} \norm{\partial_t p(t)}_{\rm HS} \le C$ which follows from  $i \partial_t p(t) = \vert h(t,x)\varphi(t) \rangle \langle  \varphi(t) \vert - \text{h.c.}$ in combination with \eqref{bound for hvarphi}. Hence we can estimate
\begin{align}
\norm{ \widetilde{K}(t) \partial_t \overline{p(t)} }_{\rm HS}^2 & = \int dk \int dx \Big\vert \int dy \widetilde K(t,k,y) \big(  (h\varphi(y)) \overline{\varphi(x)} -  \varphi(y) \overline{ h\varphi(x)}   \big) \Big\vert^2 \notag \\[1mm]
& \le 2 \norm{\varphi(t)}^2_{L^2} \norm{\widetilde K(t, \cdot ,\cdot)}_{L^2}^2 \norm{h(t) \varphi(t)}_{L^2}^2,
\end{align}
and similarly, 
\begin{align}
\norm{	( \partial_t \widetilde{K}(t))  \overline{q(t)} }^2_{\rm{HS}} =  \int dk \int dx \Big\vert \int dy (\partial_t \widetilde K(t,k,y)) \big( \delta(x-y)- \varphi(y) \overline {\varphi(x)}  \big) \Big\vert^2 .
\end{align}
In the last expression we insert $\widetilde{K}(t,k,x) = \eta(k) e^{-2\pi ikx} \varphi(t,x)$ to get
\begin{align}
\norm{	( \partial_t \widetilde{K}(t))  \overline{q(t)} }^2_{\rm{HS}} & =  \int dk \int dx \Big\vert \int dy  \, \eta(k) e^{-2\pi iky} ( \partial_t \varphi(t,y) ) \big( \delta(x-y)- \varphi(y) \overline {\varphi(x)}  \big) \Big\vert^2 \notag \\[1mm]
&  \le 4 \norm{\eta}_{L^2}^2  \norm{ h(t) \varphi(t) }_{L^2}^2.
\end{align}
Together, this implies $\sup_{t\in K}(\norm{\frac{d}{dt} A_2(t) } + \norm{\frac{d}{dt}B(t)}_{ \rm HS } ) <\infty $.
\end{proof}

\section{Proofs}
\label{sec:proofs}

\subsection{Preliminary Lemmas}

Let us start by discussing the power series expansion of the Hamiltonian $H(t)$ in more detail. We prove two lemmas in this subsection: one on the growth of the coefficients, and one on remainder estimates.
Recall that on $[0,1]$ we have the Taylor expansion
\begin{equation}
\sqrt{1-x} = \sum_{n=0}^{k} c_n x^n + \widetilde{R}_k(x),
\end{equation}
with the $c_n$ given by
\begin{equation}\label{def: coefficients cn}
c_n :=
(-1)^{n} \binom{\frac{1}{2}}{n} 
:= (-1)^{n} \frac{\frac{1}{2} \left(-\frac{1}{2}\right) \left(-\frac{3}{2}\right) \cdots \left(\frac{1}{2} - (n-1)\right)}{n!}.
\end{equation}
Defining $\widetilde{R}_k(x) := \sqrt{[1-x]_+} - \sum_{n=0}^{k} c_n x^n$ for all $x \geq 0$, an expansion of the Hamiltonian $H(t)$ yields
\begin{align}\label{Taylor_Ham}
H(t) = \sum_{\ell=0}^{r} N^{-\ell/2} H_\ell (t) +  S_N^{(r)}(t)
\end{align}
with $H_0(t)$ given by \eqref{H_Bog}, and
\begin{subequations}
\begin{align}
H_1(t) &= \int dx\, b^{*}(x) \Big( q(t) \widehat{\Phi} q(t) - \scp{\varphi(t)}{\widehat{\Phi} \varphi(t)} \Big) b(x), \label{def of H_1}\\
H_{2n}(t) & = c_{n} \int dx \int dk\,  K(t,k,x) \big( a^{*}(k) + a(-k) \big) b^*(x)  \Number_b^{n} + \hc \quad \forall n\ge 1, \label{def of H_2n} \\[1.5mm]
H_{2n+1}(t)& = 0 \quad \forall n\ge 1, \label{def of H_2n+1} \\[1.5mm]
S_N^{(0)}(t) & = N^{-1/2} H_1(t) + \Big( \int dx \int dk\,  K(t,k,x) \big( a^{*}(k) + a(-k) \big) b^*(x)    \widetilde{R}_{0}\bigg( \frac{\Number_b}{N}\bigg) + \hc \Big),\\[1.5mm]
S_N^{(1)}(t) & =  S_N^{(0)} - N^{-1/2} H_1(t),\\[1.5mm]
S_N^{(2r)}(t) & =  \int dx \int dk\,  K(t,k,x) \big( a^{*}(k) + a(-k) \big) b^*(x) \widetilde{R}_{r}\bigg( \frac{\Number_b}{N}\bigg) + \hc  \quad \forall r\ge 1,\\[1.5mm]
S_N^{(2r+1)}(t) & = S_N^{(2r)}(t) \quad \forall r\ge 1.\label{S_N even and odd}
\end{align}
\end{subequations}
Note that we write \eqref{Taylor_Ham} as an expansion in $N^{-\ell/2}$ instead of $N^{-\ell}$ due to the appearance of $H_1$.

The next lemma provides estimates that are needed to bound the above operators relative to number operators.
\begin{lemma}\label{lem_Ham_bounds} Let $n\in \mathbb N_0$ and $a^\# \in \{a,a^*\}$. Then there exist positive constants $C$ and $C(n)$, such that
\begin{align}\label{eq: aux bound for H1}
& \norm{ \int \d x \int \d k\, b^{*}(x)\eta(k) \bigg( \big( q(t) e^{-2\pi ik \cdot} q(t) - \scp{\varphi(t)}{e^{- 2\pi ik\cdot}\varphi(t)} \big)a^\#(\pm k) \bigg) b(x)  \phi } \notag \\ 
& \hspace{8cm} \le C \norm{(\mathcal N_a +1)^{1/2} (\mathcal N_b +1) \phi } 
\end{align}
and
\begin{align}\label{eq: aux bound for H2n and Sr}
\norm{ \bigg( \int dx \int dk\, K(t,k,x) a^{\#}(\pm k) b^*(x)  \Number_b^{\frac{n}{2}} + \textnormal{h.c.}\bigg) \phi } \le C(n) \norm{ (\mathcal N_a+1)^{1/2} (\mathcal N_b+1)^{\frac{n+1}{2}} \phi }
\end{align}
for all $\phi \in \mathcal F$.
\end{lemma}
\begin{proof}
To derive \eqref{eq: aux bound for H1}, let us consider the contribution from the $q(t) e^{-2\pi ik \cdot} q(t)$ term first. Straightforwardly estimating the creation and annihilation operators in terms of number operators as in \cite[Lemma~5.1]{bossmann} gives
\begin{align}
\norm{\int \d x \int \d k\, b^{*}(x) \eta(k) (q(t) e^{-2\pi ik \cdot} q(t)) a^\#(k) b(x) \phi} \leq \norm{f}_{\Hilbert \to \Hilbert\otimes\Hilbert} \norm{(\Number_a+1)^{1/2} \Number_b \phi},
\end{align}
where $\Hilbert = L^2(\mathbb R^3)$ and $f: \Hilbert \to \Hilbert \otimes \Hilbert$ is the operator with kernel
\begin{equation}
f(x,k;y) = \int \d z\, \eta(k) q(x,z) e^{ikz} q(z,y).
\end{equation}
Thus, $\norm{f}_{\Hilbert \to \Hilbert\otimes\Hilbert} \leq \norm{\eta}$. The second term inside the norm on the left side of \eqref{eq: aux bound for H1} can be estimated in complete analogy. Since \eqref{eq: aux bound for H2n and Sr} is obtained in a similar way, we omit the details (here one needs to use $\norm{K}_{\rm HS} \leq \norm{\eta}$).
\end{proof}

The remainders in Equation~\eqref{Taylor_Ham} can be estimated in terms of number operators.
\begin{lemma}\label{Taylor_H_with_remainder}
For any $r \ge 1$ there is a constant $C(r)\ge 0$, such that
\begin{align} 
\big\| S^{(r)}_N(t) \phi \big\| &\leq C(r)N^{-(r+1)/2} ~ \big\| (\Number_a+1)^{1/2} (\Number_b+1)^{(r+2)/2} \phi \big\|
\end{align}
for all $\phi \in \mathcal F$. 
\end{lemma}
\begin{proof}
For $x \in [0,1]$, Taylor's theorem gives $\sqrt{1-x} = \sum_{n=0}^{k} c_n x^n + \widetilde{R}_k(x)$,
with rest term
\begin{align}
\widetilde{R}_k(x) = (k+1)c_{k+1} \int_0^x (1-y)^{-k-1/2}(x-y)^k \d y = c_{k+1} (1-\xi)^{-k-1/2} x^{k+1} 
\end{align}
for some $\xi \in (0,x)$. For $x \in [0,1]$, we thus find that
\begin{align}\label{eq: R_k bound 1}
|\widetilde{R}_k(x)| \leq C(k) x^{k+1}.
\end{align}
For $x \geq 1$, we have $\widetilde{R}_k(x) := \sqrt{[1-x]_+} - \sum_{n=0}^{k} c_n x^n = -\sum_{n=0}^{k} c_n x^n$, therefore
\begin{align}\label{eq: R_k bound 2}
|\widetilde{R}_k(x)| \leq C(k) x^{k} ,
\end{align}
and hence $|\widetilde{R}_k(x)| \leq C(k) x^{k+1}$ for all $x\ge 0$. Moreover, combining \eqref{eq: R_k bound 1} and \eqref{eq: R_k bound 2} provides also $|\widetilde{R}_k(x)| \leq C(k) x^{k+\frac{1}{2}}$ for all $x \ge 0$.

Thus, with Lemma~\ref{lem_Ham_bounds} and $|\widetilde{R}_0(x)| \leq C(0) x^{\frac{1}{2}}$, we find
\begin{align}
\big\| S_{N}^{(0)}(t) \phi \big\|    &  \leq N^{-1/2} \norm{H_1(t) \phi}  + \norm{ \bigg( \int dx \int dk\,  K(t,k,x) \big( a^{*}(k) + a(-k) \big) b^*(x)  \widetilde{R}_0\left(\frac{\Number_b}{N}\right) + \hc \bigg) \phi} \nonumber\\
&\leq C N^{-1/2} \norm{(\Number_a+1)^{1/2} (\Number_b+1) \phi},
\end{align}
and, with $|\widetilde{R}_0(x)| \leq C(0) x$, we get
\begin{align}
\big\| S_{N}^{(1)}(t) \phi \big\| & =   \norm{ \bigg( \int dx \int dk\, K(t,k,x) \big( a^{*}(k) + a(-k) \big) b^*(x)  \widetilde{R}_0\left(\frac{\Number_b}{N}\right) + \hc \bigg) \phi} \nonumber\\
&\leq C N^{-1} \norm{(\Number_a+1)^{1/2} (\Number_b+1)^{3/2} \phi}.
\end{align}
Then, for $r$ even, we use Lemma~\ref{lem_Ham_bounds} and $|\widetilde{R}_k(x)| \leq C(k) x^{k+\frac{1}{2}}$, and find
\begin{align}
\norm { S_N^{(r)}(t) }  & = \norm{ \bigg( \int dx \int dk\,  K(t,k,x) \big( a^{*}(k) + a(-k) \big)  b^*(x) \widetilde{R}_{r/2}\left(\frac{\Number_b}{N} \right)  + \hc \bigg) \phi} \nonumber\\
&\leq N^{- (r+1) / 2} C(r) \norm{(\Number_a+1)^{1/2} (\Number_b+1)^{(r+2)/2} \phi}.
\end{align}
For $r$ odd, we employ $S_N^{(r)} = S^{(r-1)}_N$ and $|\widetilde{R}_k(x)| \leq C(k) x^{ k + 1 }$ to obtain
\begin{align}
\norm{S_{N}^{(r)}(t) \phi} = \norm{S_{N}^{(r-1)}(t) \phi} & = \norm{ \bigg( \int dx \int dk\, K(t,k,x) \big( a^{*}(k) + a(-k) \big)  b^*(x) \widetilde{R}_{(r-1)/2}\left(\frac{\Number_b}{N} \right)  + \hc \bigg) \phi} \nonumber\\
&\leq N^{-(r+1)/2}C((r+1)/2) \big\| (\Number_a+1)^{1/2} (\Number_b+1)^{(r+2)/2} \phi \big\|.
\end{align}

\end{proof}

\subsection{Propagation of Moments of Number Operators}

We now prove that moments of the number operator $\Number = \Number_a + \Number_b$ with respect to the state $\chi_{\ell}(0)$ can be propagated in time.

\begin{lemma}\label{lem_propagation_of_moments}
Let Assumption~\ref{main_number_op_assumption} hold. Then for all $n \in \mathbb N_0$ and $\ell \in \{0,...,r\}$ there is a $C(n,\ell) > 0$ such that
\begin{equation}
\norm{(\Number +1)^{n} \chi_{\ell}(t)} \leq C(n,\ell) e^{C(n,\ell)t}.
\end{equation}
\end{lemma}

\begin{proof}
By the definition of $\chi_{\ell}(t)$ from \eqref{chil_def},
\begin{align}\label{number_op_propagation_terms}
&\norm{(\Number + 1 )^{n} \chi_{\ell}(t)} \nonumber\\
&\quad\leq \norm{ (\Number +1 )^{n} U_0(t,0) \chi_{\ell}(0)} \nonumber\\
&\qquad + \sum_{m=0}^{\ell-1} \sum_{k=1}^{\ell-m} \sum_{\substack{\alpha \in \mathbb{N}^k \\ |\alpha| = \ell - m}} \int_{\Delta_k} ds^{(k)} \norm{ ( \Number + 1 )^{n} \prod_{i=1}^k \widetilde{H}_{\alpha_i }(s_i,t) U_{0}(t,0) \chi_m(0)}.
\end{align}
Then by the number operator bound in \eqref{eq:4x} and by Assumption~\ref{main_number_op_assumption},
\begin{align}
\norm{(\Number +1)^{n} U_{0}(t,0) ]\chi_{\ell}(0)} \leq C(n,\ell) e^{C(n)t}.
\end{align}
Next, let us estimate the term
\begin{equation}\label{term_number_est}
\norm{(\Number +1)^{n} \widetilde{H}_{2j}(s,t) U_{0}(t,0) \phi} = \norm{ (\Number +1)^{n} U_{0}(t,s) H_{2j}(s) U_{0}(s,0) \phi}
\end{equation}
for any $j=\frac{1}{2}$ or $j\in\NNN$, and $\phi \in \Fock$. In the following, we use \eqref{eq:4x} in the first step, then commute $H_{2j}(s)$ with the number operators and use Lemma~\ref{lem_Ham_bounds} in the second step, and use \eqref{eq:4x} again in the third step; we find
\begin{align}
\eqref{term_number_est} &\leq e^{C(n)|t-s|} \norm{ ( \Number + 1 )^{n} H_{2j}(s) U_{0}(s,0) \phi} \nonumber\\
&\leq C(n ,j) e^{C(n)|t-s|} \norm{ ( \Number + 1 )^{n+j+1}  U_{0}(s,0) \phi} \nonumber\\
&\leq C(n ,j) e^{C(n)|t-s|}e^{C(n,j)|s|} \norm{(\Number + 1 )^{n+ j +1} \phi}.
\end{align}
Finally, note that the term with the highest number of creation and annihilation operators in the last line of \eqref{number_op_propagation_terms} for given $m$ comes from $k=\ell-m$, i.e., $\alpha = (1,1,\ldots,1)$. Thus, 
\begin{align}
&\sum_{m=0}^{\ell-1} \sum_{k=1}^{\ell-m} \sum_{\substack{\alpha \in \mathbb{N}^k \\ |\alpha| = \ell - m}} \int_{\Delta_k} ds^{(k)} \norm{(\Number + 1 )^{n}  \prod_{i=1}^k \widetilde{H}_{\alpha_i }(s_i,t) U_{0}(t,0) \chi_m(0)} \nonumber\\
&\quad \leq \sum_{m=0}^{\ell-1} C(n,m) e^{C(n,m)t} \norm{(\Number + 1 )^{n + 3(\ell-m)/2} \chi_m(0)},
\end{align}
and the lemma is proven by Assumption~\ref{main_number_op_assumption}.
\end{proof}

\subsection{Proof of the Theorems}

\begin{proof}[Proof of Theorems~\ref{thm_Bog} and \ref{thm_higher_orders}]
We define the difference
\begin{equation}
\chi^{\mathrm{rest}}_r(t) : = \chi_{\leq N}(t) - \sum_{\ell=0}^r N^{-\ell/2} \chi_{\ell}(t) \in \Gock,
\end{equation}
where we extend the state $\chi_{\le N}\in \Gock_{\le N}$ to a state in $\Gock$ by setting $\chi_{\le N}^{(k)}=0$ for all $k\ge N+1$. Then
\begin{equation}
\norm{\Psi_N(t) - \Psi_N^{(r)}(t)} = \norm{\chi_{\leq N}(t) - \sum_{\ell=0}^r N^{-\ell/2} \chi_{\ell}(t)}_{\Gock_{\leq N}} = \norm{ \chi^{\mathrm{rest}}_r (t)}_{\Gock_{\leq N}},
\end{equation}
where $\|\phi\|_{\Gock_{\leq N}} := \|\phi|_{\Gock_{\leq N}}\|$  (which defines only a semi-norm). Now note that
\begin{align}
\chi^{\mathrm{rest}}_r (t) = U_0(t,0) \chi^{\mathrm{rest}}_r  (0) - i \int_0^t \d s\, U_{0}(t,s) F(s),
\end{align}
with
\begin{align}
F(s) & :=\big(H^{\leq N}(s) - H_0(s)\big) \chi_{\leq N}(s) - \sum_{\ell=0}^r N^{-\ell/2} \sum_{m=1}^\ell H_{m}(s) \chi_{\ell-m}(s) \nonumber\\
& = \big(H^{\leq N}(s) - H_0(s)\big)  \chi^{\mathrm{rest}}_r  (s) + \sum_{\ell=0}^r N^{-\ell/2} \left( \big(H^{\leq N}(s) - H_0 (s)\big) \chi_\ell(s) - \sum_{m=1}^\ell H_m (s) \chi_{\ell-m}(s) \right) \nonumber\\
& = \big(H^{\leq N}(s) - H_0(s)\big)  \chi^{\mathrm{rest}}_r(s) + \sum_{\ell=0}^r N^{-\ell/2} \left( H^{\leq N}(s) - \sum_{m=0}^{r-\ell} N^{-m/2}H_m (s) \right) \chi_\ell(s),
\end{align}
where we reordered the summation in the last step. Then
\begin{align}\label{rest_comp_1}
&\norm{ \chi^{\mathrm{rest}}_r (t)}^2_{\Gock_{\leq N}} \nonumber\\
& = \norm{\chi^{\mathrm{rest}}_r  (0)}^2_{\Gock_{\leq N}} + 2\Im \int_0^t \d s\, \SCP{U_0 (s,0) \chi^{\mathrm{rest}}_r  (0)}{F(s)}_{\Gock_{\leq N}} + \int_0^t \d s \int_0^t \d \tilde s \, \SCP{U_0(0,\tilde s )F(\tilde s )}{U_0(0,s)F(s)}_{\Gock_{\leq N}} \nonumber\\
& = \norm{\chi^{\mathrm{rest}}_r(0)}^2_{\Gock_{\leq N}} + 2 \Im \int_0^t \d s\, \SCP{\chi^{\mathrm{rest}}_r(s)}{F(s)}_{\Gock^{\leq N}} - 2\Re \int_0^t \d s \int_0^s \d \tilde s \, \SCP{U_ 0(0,\tilde s )F(\tilde s )}{U_0 (0,s)F(s)}_{\Gock_{\leq N}} \nonumber\\
& \quad + \int_0^t \d s \int_0^t \d \tilde s \, \SCP{U_0(0,\tilde s)F(\tilde s )}{U_0(0,s)F(s)}_{\Gock_{\leq N}} \nonumber\\
&= \norm{\chi^{\mathrm{rest}}_r(0)}^2_{\Gock_{\leq N}} + 2 \Im \int_0^t \d s\, \SCP{\chi^{\mathrm{rest}}_r(s)}{F(s)}_{\Gock_{\leq N}} \nonumber\\
&= \norm{\chi^{\mathrm{rest}}_r(0)}^2_{\Gock_{\leq N}} + 2 \sum_{\ell=0}^r N^{-\ell/2} \Im \int_0^t \d s\, \SCP{\chi^{\mathrm{rest}}_r(s)}{\left( H^{\leq N}(s) - \sum_{m=0}^{r-\ell} N^{-m/2} H_m (s) \right) \chi_\ell(s)}_{\Gock_{\leq N}},
\end{align}
using self-adjointness of $H^{\leq N}(s) - H^{(0)}(s)$ in the last step. By the definition \eqref{Taylor_Ham} of the rest term and using the Cauchy--Schwarz inequality, we find
\begin{align}\label{rest_comp_2}
\norm{\chi^{\mathrm{rest}}_r(t)}^2_{\Gock_{\leq N}} &= \norm{\chi^{\mathrm{rest}}_r(0)}^2_{\Gock_{\leq N}} + 2 \sum_{\ell=0}^r N^{-\ell / 2 } \Im \int_0^t \d s\, \SCP{\chi^{\mathrm{rest}}_r(s)}{S^{(r-\ell)}(s) \chi_\ell(s)}_{\Gock_{\leq N}} \nonumber\\
&\leq \norm{\chi^{\mathrm{rest}}_r(0)}^2_{\Gock_{\leq N}} + 2 \sum_{\ell=0}^r N^{- \ell /2  } \int_0^t \d s\, \norm{\chi^{\mathrm{rest}}_r(s)}_{\Gock_{\leq N}} \norm{S^{(r-\ell)}(s) \chi_\ell(s)}_{\Gock_{\leq N}}.
\end{align}
Now we use the estimate of the remainder in terms of number operators from Lemma~\ref{Taylor_H_with_remainder}, and the estimate of moments of number operators from Lemma~\ref{lem_propagation_of_moments}, which yield
\begin{align}\label{proof_conclusion_1}
&\norm{\chi^{\mathrm{rest}}_r(t)}^2_{\Gock_{\leq N}} - \norm{\chi^{\mathrm{rest}}_r(0)}^2_{\Gock_{\leq N}} \nonumber\\
&\quad\leq 2 N^{-(r+1)/2} \int_0^t \d s\, \norm{\chi^{\mathrm{rest}}_r(s)}_{\Gock_{\leq N}} \sum_{\ell=0}^r C(r-\ell) \norm{(\Number_a+1)^{1/2} (\Number_b+1)^{(r-\ell+2)/2} \chi_{\ell}(s)} \nonumber\\
&\quad\leq N^{-(r+1)/2} \int_0^t \d s\, \norm{\chi^{\mathrm{rest}}_r(s)}_{\Gock_{\leq N}} C(r) e^{C(r)s}  \nonumber\\
&\quad\leq \int_0^t \d s\, \left( \frac{1}{2} N^{-r-1} C(r)^2 e^{2C(r)s} + \frac{1}{2} \norm{\chi^{\mathrm{rest}}_r(s)}^2_{\Gock_{\leq N}} \right).
\end{align}
Then Gronwall's lemma implies
\begin{align}\label{proof_conclusion_2}
\norm{\chi^{\mathrm{rest}}_r(t)}^2_{\Gock_{\leq N}} \leq C(r) e^{C(r) t} \left( \norm{\chi^{\mathrm{rest}}_r(0)}^2_{\Gock_{\leq N}} + N^{-r-1} \right).
\end{align}
\end{proof}

\begin{proof}[Proof of Theorem~\ref{thm_higher_orders_unitary_group}]
We abbreviate
\begin{align}
\widetilde{\chi}^{\mathrm{rest}}_r(t,s) = U(t,s) \chi - \sum_{\ell = 0}^r N^{-\ell/2} U_{\ell}(t,s) \chi.
\end{align}
A computation analogous to \eqref{rest_comp_1} and \eqref{rest_comp_2} yields
\begin{gather}
\begin{align}
\norm{\widetilde{\chi}^{\mathrm{rest}}_r(t,s)}^2 - \norm{\widetilde{\chi}^{\mathrm{rest}}_r( t_0 ,s)}^2 &= 2 \sum_{\ell=0}^r N^{- \ell /2 } \Im \int_{t_0}^t \d \tilde{s}\, \SCP{\widetilde{\chi}^{\mathrm{rest}}_r(\tilde{s},s)}{S_N^{(r-\ell)}(\tilde{s}) U_{\ell}(\tilde{s},s) \chi} \nonumber\\
&\leq 2 \sum_{\ell=0}^r N^{-\ell /2 } \int_{ t_0 }^t \d \tilde{s}\, \norm{\widetilde{\chi}^{\mathrm{rest}}_r(\tilde{s},s)} \norm{S_N^{(r-\ell)}(\tilde{s}) U_{\ell}(\tilde{s},s) \chi}.
\end{align}
\end{gather}
The rest term is bounded in terms of number operators according to Lemma~\ref{Taylor_H_with_remainder}. Furthermore, by the same computations as in the proof of Lemma~\ref{lem_propagation_of_moments}, we deduce that
\begin{equation}
\norm{(\Number +1)^{n}  U_{\ell}(t,s) \chi} \leq C(n,\ell) e^{C(n, \ell)(|t|+|s|)},
\end{equation}
for all $n \in \NNN_0$ and for $\chi$ satisfying our assumption \eqref{ass_number_unitary}. Then the proof is concluded as in \eqref{proof_conclusion_1} and \eqref{proof_conclusion_2} in the proof of Theorem~\ref{thm_higher_orders}.
\end{proof}

\appendix

\section{More details on the excitation Fock spaces}
\label{section: more details on the excitation Fock space}
In this appendix we give further details about the unitary, defined by \eqref{decomp_Psi}, and the derivation of the excitation Hamiltonian from \eqref{full_H}. To this end, we closely follow \cite[Chapter 4.1]{Lewin:2015a} and \cite[Chapter 2.3]{LNSS2015}.
First note that the unitary Weyl operator $W^{*} \left( \sqrt{N} \alpha(t) \right)$ maps the Fock space $\mathcal{F} = \bigoplus_{n=0}^{\infty} (L^2(\mathbb{R}^3))^{\otimes_s n}$ into itself. Under this mapping the coherent state $W(\sqrt{N} \alpha(t)) \ket{\Omega}$ is mapped to the vacuum vector. 
Second we recall that any function $\Psi \in (L^2(\mathbb{R}^3))^{\otimes_s N}$ can be written as
\begin{align}
\Psi
&= \psi^{(0)} \varphi(t)^{\otimes N} 
+ \psi^{(1)} \otimes_s \varphi(t)^{\otimes(N-1)}
+ \psi^{(2)} \otimes_s \varphi(t)^{\otimes (N-2)} + \ldots + \psi^{(N)}.
\end{align}
One way to obtain this decomposition is to introduce for $k \in \{0, 1, \ldots, N \}$ the operators
\begin{align}
P_{N,k} 
&= \sum_{\substack{a \in \{0,1 \}^N \\ \sum_{i} a_i = k}} \prod_{i = 1}^{N} p_i^{1 - a_i} q_i^{a_i}
= \frac{1}{k! (N-k)!}  \sum_{\sigma \in S_N}  q_{\sigma(1)} \cdots q_{\sigma(k)} 
p_{\sigma(k+1)} \cdots p_{\sigma(N)} 
\end{align}
with $p_i(t) = \ket{\varphi(t)} \bra{\varphi(t)}_i$ and $q_i(t) = 1 - p_i(t)$
satisfying the identity $\sum_{k =0}^{N} P_{N,k} = \id_{L^2(\mathbb{R}^{3N})}$ (see, e.g., \cite[Section 3.3.1]{KnowlesP2009}). Then
$\Psi = \sum_{k=0}^N P_{N,k} \Psi$ where $P_{N,k} \Psi = \psi^{(k)} \otimes_s \varphi^{\otimes (N-k)}$ with
\begin{align}
\psi^{(k)}
&= \binom{N}{k}^{1/2}  \prod_{i=1}^k q_i  \scp{\varphi^{\otimes (N-k)}}{\Psi}_{L^2(\mathbb{R}^{3(N-k)})} .
\end{align}
The mapping $\widetilde{U}_N(t): (L^2(\mathbb{R}^3))^{\otimes_s N} \rightarrow \bigoplus_{k=0}^N ( \varphi(t)^\perp)^{\otimes_s k}$, $\Psi \mapsto \left( \psi^{(k)}(t) \right)_{k=0}^N$ is unitary because 
\begin{align}
\left< \psi^{(k)} \otimes_s \varphi(t)^{\otimes (N-k)} , \psi^{(l)} \otimes_s \varphi(t)^{\otimes (N-l)} \right>_{L^2(\mathbb{R}^{3N})}
= \delta_{k,l} \left< \psi^{(k)} , \psi^{(l)} \right>_{L^2(\mathbb{R}^{3k})} 
\end{align}
and therefore
$\norm{\Psi}_{L^2(\mathbb{R}^{3N})}^2 = \left| \psi^{(0)} \right|^2 + \sum_{k=1}^N \norm{\psi^{(k)}}^2_{L^2(\mathbb{R}^{3k})}$.

Let $U_N(t): \mathcal{H}_N \rightarrow \left( \bigoplus_{k=0}^N ( \varphi(t)^\perp)^{\otimes_s k} \right) \otimes \mathcal{F}$ be the unitary operator defined by $U_N(t)= \widetilde{U}_N(t) \otimes W^{*} \left( \sqrt{N} \alpha(t) \right)$. In analogy to \cite[Chapter 4]{LNSS2015} and \cite[Chapter 4.1]{Lewin:2015a} we can use the inclusions 
$ \mathcal{H}_N = (L^2(\mathbb{R}^3))^{\otimes_s N} \otimes \mathcal{F} \subset \mathcal{F} \otimes \mathcal{F}$ and
$\left( \bigoplus_{k=0}^N ( \varphi(t)^\perp)^{\otimes_s k} \right) \otimes \mathcal{F}_a \subset \mathcal{F} \otimes \mathcal{F}$ to represent $U_N(t)$ and its adjoint in terms of annihilation and creation operators. To this end,  we use $b(f) \otimes \id_{\mathcal{F}}, b^{*}(f) \otimes \id_{\mathcal{F}}, \mathcal{N}_b \otimes \id_{\mathcal{F}}$ and $\id_{\mathcal{F}} \otimes a(f), \id_{\mathcal{F}} \otimes a^{*}(f), \id_{\mathcal{F}} \otimes \mathcal{N}_a$ with $f \in L^2(\mathbb{R}^3)$ to denote the usual annihilation, creation and number of particles operators on the first and second Fock space of $\mathcal{F} \otimes \mathcal{F}$.
\begin{lemma}
\label{lemma:properties U}
The operators $U_N(t)$ and $U_N(t)^*$ can equivalently be written as 
\begin{align}
U_N(t) \Psi_N &= W^{*} \left( \sqrt{N} \alpha(t) \right) \bigoplus_{k=0}^N q(t)^{\otimes k} \frac{b(\varphi(t))^{N- k}}{\sqrt{(N- k)!}} \Psi_N ,
\\
U_N(t)^* \left( \bigoplus_{k=0}^N \chi_{\leq N}^{(k)} \right)
&= W \left( \sqrt{N} \alpha(t) \right) \sum_{k=0}^N \frac{b^{*} (\varphi(t))^{N-k}}{\sqrt{(N-k)!}} \chi_{\leq N}^{(k)}
\end{align}
for all $\Psi_N \in\mathcal{H}_N$ and $ \chi_{\leq N}^{(k)} \in  ( \varphi(t)^\perp)^{\otimes_s k} \otimes \mathcal{F}$, $k = 0, \ldots, N$. 
On $\left( \bigoplus_{k=0}^N ( \varphi(t)^\perp)^{\otimes_s k} \right) \otimes \mathcal{F}$ we have
\begin{subequations}
\begin{align}
U_N(t) b^{*}(\varphi(t)) b(\varphi(t) ) U_N(t)^*  & = N - \left(\mathcal{N}_b(t) \right)_{+}  ,
\\
U_N(t) b^{*}(f) b(\varphi(t) ) U_N(t)^*  & = b^{*}(f) \left[ N - \left(\mathcal{N}_b(t) \right)_{+} \right]^{1/2}  ,
\\
U_N(t) b^{*}(\varphi(t)) b(f) U_N(t)^*  & =  \left[N - \left(\mathcal{N}_b(t) \right)_{+} \right]^{1/2} b(f)  ,
\\
U_N(t) b^{*}(f) b(g) U_N(t)^*  & =  b^{*}(f) b(g)  ,
\end{align}
\end{subequations}
for all $f,g \in \{\varphi(t)\}^{\perp}$ and $\left(\mathcal{N}_b(t) \right)_{+} = \mathcal{N}_b - b^{*}(\varphi(t)) b(\varphi(t))$.
Moreover,
\begin{subequations}
\begin{align}
\label{eq: action unitary mapping on a}
U_N(t) a(h) U_{N}(t)^* = a(h) + \sqrt{N} \scp{h}{\alpha(t)}_{L^2(\mathbb{R}^3)} ,
\\
\label{eq: action unitary mapping on a dagger}
U_N(t) a^{*}(h) U_{N}(t)^* = a^{*}(h) + \sqrt{N} \scp{\alpha(t)}{h}_{L^2(\mathbb{R}^3)},
\end{align}
\end{subequations}
for all $h \in L^2(\mathbb{R}^3)$.
\end{lemma}

\begin{proof}
The first part of the Lemma is a direct consequence of $\left[ \widetilde{U}_N(t) \otimes \id_{\mathcal{F}} , \id_{\mathcal{F}} \otimes W^{*} \left( \sqrt{N} \alpha(t) \right) \right] = 0$ and \cite[Proposition 4.2]{LNSS2015}. The relations \eqref{eq: action unitary mapping on a} and \eqref{eq: action unitary mapping on a dagger} follow from the shifting property of the Weyl operators.
\end{proof}

In the following, we briefly explain how one derives the Schr\"odinger equation \eqref{Schr_eq_Fock}. To this end, we need to compute the time derivative of the unitary mapping.

\begin{lemma}
\label{lemma: time evolution unitary}
Let $(\varphi(t), \alpha(t))$ be a sufficiently regular trajectory on $L^2(\mathbb{R}^3) \oplus L^2(\mathbb{R}^3)$ satisfying $\norm{\varphi(t)}_{L^2(\mathbb{R}^3)} = \norm{\varphi(0)}_{L^2(\mathbb{R}^3)}$ for all $t \geq 0$. Then the time derivative of $U_N(t)$ is
\begin{align}
\begin{split}
i \dot{U}_N(t)
&= \bigg[ b^{*} \left( \varphi(t) \right) b \left( q(t) i \dot{\varphi}(t) \right) 
- \sqrt{N - \left(\mathcal{N}_b(t) \right)_{+}} b \left( q(t) i \dot{\varphi}(t) \right)
\\
&\quad - b^{*} \left( q(t) i \dot{\varphi}(t) \right) \sqrt{N - \left(\mathcal{N}_b(t) \right)_{+}}
- \scp{i \dot{\varphi}(t)}{\varphi(t)} \left( N - \left(\mathcal{N}_b(t) \right)_{+} \right)
 \\
&\quad - N \, \Re \scp{i \dot{\alpha}(t)}{\alpha(t)}
- \sqrt{N} a \left( i \dot{\alpha}(t) \right)
- \sqrt{N} a^{*} \left( i \dot{\alpha}(t) \right)
\bigg] U_N(t) .
\end{split}
\end{align}
\end{lemma}

\begin{proof}
From \cite[Lemma 6]{Lewin:2015a} we know that 
\begin{align}
\begin{split}
i \dot{\widetilde{U}}_N(t) \otimes \id_{\mathcal{F}}
&= \bigg[ b^{*} \left( \varphi(t) \right) b \left( q(t) i \dot{\varphi}(t) \right) 
- \sqrt{N - \left(\mathcal{N}_b(t) \right)_{+}} b \left( q(t) i \dot{\varphi}(t) \right) \\
&\quad - b^{*} \left( q(t) i \dot{\varphi}(t) \right) \sqrt{N - \left(\mathcal{N}_b(t) \right)_{+}}
- \scp{i \dot{\varphi}(t)}{\varphi(t)} \left( N - \left(\mathcal{N}_b(t) \right)_{+} \right) \bigg] \widetilde{U}_N(t) \otimes \id_{\mathcal{F}} .
\end{split}
\end{align}
Combining this with
\begin{align}
  \id_{\mathcal{F}} \otimes i \dot{W}^{*} \left( \sqrt{N} \alpha(t) \right)
&= - \bigg[ N \, \Re \scp{i \dot{\alpha}(t)}{\alpha(t)}
+ \sqrt{N} a \left( i \dot{\alpha}(t) \right)
+ \sqrt{N} a^{*} \left( i \dot{\alpha}(t) \right)
\bigg]  \id_{\mathcal{F}} \otimes W^{*} \left( \sqrt{N} \alpha(t) \right)
\end{align}
and using that both unitary mappings commute shows the claim.
The last equality can be obtained by  \cite[Lemma A.3. ($\alpha= 1$, $a = b$ and $a^{*} = b^{*}$)]{FG2017} and the fact that $W^{*}(f) = W(-f)$ for $f \in L^2({\mathbb{R}^3})$.
\end{proof}
Now, let $\Psi_N(t)$ and $(\varphi(t), \alpha(t))$ be solutions of \eqref{Schr_eq} and \eqref{varphi_eq}--\eqref{alpha_eq} such that $\norm{\varphi(0)}_{L^2(\mathbb{R}^3)} = 1$. Then $\chi_{\le N}(t) \in \left( \bigoplus_{k=0}^N ( \varphi(t)^\perp)^{\otimes_s k} \right) \otimes \mathcal{F} \subset \mathcal{F} \otimes \mathcal{F}$, given by $\chi_{\le N}(t) = U(t) \Psi_{N}(t)$ satisfies
\begin{align}
i \dot{\chi}_{\le N}(t)
&= \left( U_N(t) H_N^{\mathrm{Nelson}} U_N(t)^* + i \dot{U}_N(t) U_N(t)^* \right)  \chi_{\le N}(t) .
\end{align}
The Schr\"odinger equation \eqref{Schr_eq_Fock} is then obtained by means of Lemma \ref{lemma: time evolution unitary}, the Schr\"odinger--Klein--Gordon equations \eqref{varphi_eq}--\eqref{alpha_eq}, 
\begin{align}
b^{*}(\varphi(t)) b(q(t) h(t) \varphi(t))
&= \int dx \, b^{*}(x) 
\Big( h(t) - h(t) p(t) - q(t) h(t) q(t) \Big) b(x)
\end{align}
and 
\begin{align}\label{Ultimate_Relation}
\begin{split}
U_N(t) H_N^{\mathrm{Nelson}} U_N(t)^*
&= N \norm{\sqrt{w} \alpha(t)}_2^2 + H_f + N^{1/2}
\left( a(\omega \alpha(t) )  + a^{*}(\omega \alpha(t) )  \right)
\\
&\quad + \left( \scp{\varphi(t)}{h(t) \varphi(t)} + \mu(t) \right) 
\left( N - \left(\mathcal{N}_b(t) \right)_{+} \right)
\\
&\quad + \int dx \, b^{*}(x) q(t) \left( - \Delta + \Phi(t, \cdot) \right) q(t) b(x) 
\\
&\quad + N^{-1/2} \int dx\, b^{*}(x) \Big( q(t) \widehat{\Phi} q(t) - \scp{\varphi(t)}{\widehat{\Phi} \varphi(t)} \Big) b(x)
\\
&\quad +
N^{-1/2} \scp{\varphi(t)}{\widehat{\Phi} \varphi(t)} \left( N - \left(\mathcal{N}_b(t) \right)_{+} + \mathcal{N}_b \right)
 \\
&\quad
+ \Big( b^{*} \left( q(t) (- \Delta + \Phi(t,\cdot)) \varphi(t) \right) \left[ N - \left(\mathcal{N}_b(t) \right)_{+} \right]^{1/2} +  \text{h.c.} \Big) 
\\
&\quad + \int dx \int dk\,  K(t,k,x) \big( a^{*}(k) + a(-k)\big) b^*(x) \big[ 1- N^{-1}  \left(\mathcal{N}_b(t) \right)_{+} \big]^{1/2} + \text{h.c.} \, .
\end{split}
\end{align}
In order to obtain Equation~\eqref{Ultimate_Relation} one needs to write $H_N^{\mathrm{Nelson}}$ in the second quantized form on $\mathcal{F} \otimes \mathcal{F}$ and proceed in a similar way as in \cite[Chapter 4]{LNSS2015} and \cite[Appendix B]{Lewin:2015a}. Moreover, note that $\left( \mathcal{N}_b(t) \right)_{+} \chi_{\le N} = \mathcal{N}_b \chi_{\le N}$ holds for all $\chi_{\le N} \in \left( \bigoplus_{k=0}^N ( \varphi(t)^\perp)^{\otimes_s k} \right) \otimes \mathcal{F}$.
For notational convenience, we define the Fock spaces $\mathcal{F}_a$, $\mathcal{F}_b$ and $\mathcal G_{\leq N}$ as in \eqref{eq: definition excitation Fock space particles}, \eqref{eq: definition excitation Fock space field bosons} and \eqref{eq: definition k-particles excitation space} and view $U_N(t)$ as a mapping from $\mathcal{H}_N$ to $\mathcal G_{\leq N}$.\\

\noindent{\it Acknowledgments.} We would like to thank Phan Th\` anh Nam for
helpful remarks regarding the well-posedness of the Bogoliubov time
evolution.  N.L.\ gratefully acknowledges support from the SNSF Eccellenza
project PCEFP2 181153 and the NCCR SwissMAP. M.F.\ has been partially
supported by the European Research Council (ERC) under the European Union’s
Horizon 2020 research and innovation programme (ERC CoG UniCoSM, grant
agreement n.724939).

{}


\begin{thebibliography}{11}
  \addcontentsline{toc}{chapter}{Bibliography}

\bibitem{A2000}
Z.~Ammari.
Asymptotic completeness for a renormalized nonrelativistic Hamiltonian in quantum field theory: the Nelson model.
\emph{Math. Phys. Anal. Geom.} 3(3), 217--285 (2000).  

\bibitem{AF2014}
Z.~Ammari and M.~Falconi.
Wigner measures approach to the classical limit of the Nelson model:
convergence of dynamics and ground state energy.
\emph{J. Stat. Phys.} 157(2), 330--362 (2014).

\bibitem{AF2017}
Z.~Ammari and M.~Falconi.
Bohr's correspondence principle for the renormalized Nelson model.
\emph{SIAM J. Math. Anal.} 49(6), 5031--5095 (2017).

\bibitem{Boccato:2016} C.~Boccato, S.~Cenatiempo, and B.~Schlein. Quantum many-body fluctuations around nonlinear Schr{\"o}dinger dynamics. \emph{Ann.\ Henri Poincar{\'e}}, 18(1):113--191 (2016).

\bibitem{bossmann} L.~Bo{\ss}mann, S.~Petrat, P.~Pickl, and A.~Soffer. Beyond Bogoliubov dynamics. \emph{Pure Appl. Anal.}, to appear. \emph{Preprint,} \href{https://arxiv.org/pdf1912.11004}{arXiv:1912.11004} (2019).

\bibitem{BPS2021} L.~Bo{\ss}mann, S.~Petrat and R.~Seiringer. Asymptotic expansion of low-energy excitations for weakly interacting bosons. \emph{Forum Math. Sigma} 9, E28 (2021).

\bibitem{BR1981}
O.~Bratteli and D.~W.~Robinson. Operator algebras and quantum-statistical mechanics. II. Equilibrium states. Models in quantum-statistical mechanics. Texts and Monographs in Physics. \emph{Springer-Verlag, New York-Berlin} (1981).

\bibitem{brennecke:2017} C.~Brennecke, P.\,T.~Nam, M.~Napi{\'o}rkowski and B.~Schlein. Fluctuations of N-particle quantum dynamics around the nonlinear Schr\"odinger equation. \emph{Ann.\ Inst.\ H.\ Poincar{\'e} Anal.\ Non Lin{\'e}aire} 36(5), 1201--1235 (2019).

\bibitem{CCFO2019} R.~Carlone, M.~Correggi, M.~Falconi and M.~Olivieri.
  Microscopic derivation of time-dependent point interactions. \emph{SIAM
    J. Math. Anal.} 53(4), 4657-4691 (2021).

\bibitem{Chong2016} J.~Chong. Derivation of large boson systems with attractive interaction and a derivation of the cubic focusing NLS in $\mathbb{R}^3$. \emph{Preprint} \href{https://arxiv.org/abs/1608.01615}{arXiv:1608.01615} (2016).

\bibitem{CHT2018}
J.~Colliander, J.~Holmer and N.~Tzirakis.
Low regularity global well-posedness for the {Z}akharov and {K}lein--{G}ordon--{S}chr\"odinger systems.
\emph{Trans. Amer. Math. Soc.} 360(9), 4619--4638 (2008).

\bibitem{CF2018} M.~Correggi and M.~Falconi. Effective potentials generated by field interaction in the quasi-classical  limit.
\emph{Ann. Henri Poincar\'e} 19(1), 189-235 (2018).
  
\bibitem{CFO2019} M.~Correggi, M.~Falconi and M.~Olivieri.  Quasi-classical
  dynamics.  \emph{J. Eur. Math. Soc.}, to appear. \emph{Preprint},
  \href{https://arxiv.org/abs/1909.13313}{arXiv:1909.13313} (2019).

\bibitem{CFO2020} M.~Correggi, M.~Falconi and M.~Olivieri.  Ground state
  properties in the quasi-classical regime. \emph{Preprint},
  \href{https://arxiv.org/abs/2007.09442}{arXiv:2007.09442} (2020).

\bibitem{davies}
E.\,B.~Davies.
Particle-boson interactions and the weak coupling limit.
\emph{J. Math. Phys.} 20, 345--351 (1979).

\bibitem{DG1999}
J.~Derezi\'nski and C.~G\'erard. Asymptotic completeness in quantum field
theory. Massive Pauli--Fierz Hamiltonians. \emph{Rev. Math. Phys.} 11(4),
383--450 (1999). 

\bibitem{DG2013}
J.~Derezi\'nski and C.~G\'erard. Mathematics of Quantization and Quantum Fields (Cambridge Monographs on Mathematical Physics). Cambridge: Cambridge University Press (2013).

\bibitem{falconi}
M.~Falconi.
Classical limit of the Nelson model with cutoff.
\emph{J. Math. Phys.} 54(1), 012303 (2013).

\bibitem{FRS2021} D.~Feliciangeli, S.~Rademacher and R.~Seiringer. Persistence of the spectral gap for the Landau--Pekar equations. \emph{Lett. Math. Phys.} 111(19) (2021).

\bibitem{FG2017}
R.\,L.~Frank and Z.~Gang.
Derivation of an effective evolution equation for a strongly coupled polaron.
\emph{Anal. PDE} 10(2), 379--422 (2017).

\bibitem{FS2014}
R.\,L.~Frank and B.~Schlein.
Dynamics of a strongly coupled polaron.
\emph{Lett. Math. Phys.} 104, 911--929 (2014).

\bibitem{GV1979a}
J.~Ginibre and G.~Velo. The classical field limit of scattering theory for non-relativistic many-boson systems. I..
\emph{Commun. Math. Phys.} 66(1), 37--76 (1979).

\bibitem{GV1979b}
J.~Ginibre and G.~Velo. The classical field limit of scattering theory for non-relativistic many-boson systems. II..
\emph{Commun. Math. Phys.} 68(1), 45--68 (1979).

\bibitem{Ginbre_Velo_expansion2}
J.~Ginibre and G.~Velo. The classical field limit of non-relativistic bosons. II. Asymptotic
expansions for general potentials.
\emph{Ann. Inst. H. Poincaré Physique théorique} 33(4), 363--394 (1980).

\bibitem{Ginbre_Velo_expansion1}
J.~Ginibre and G.~Velo. The classical field limit of nonrelativistic bosons. I. Borel summa-
bility for bounded potentials.
\emph{Ann. Phys.} 128(2), 243--285 (1980).

\bibitem{G2017}
M.~Griesemer.
On the dynamics of polarons in the strong-coupling limit.
\emph{Rev. Math. Phys.} 29(10), 1750030 (2017).

\bibitem{GW2018}
M.~Griesemer and A.~W\"{u}nsch.
On the domain of the Nelson Hamiltonian.
\emph{J. Math. Phys.} 59(4), 042111 (2018).

\bibitem{GrillakisMachedon:2013} M.~Grillakis and M.~Machedon. Pair excitations and the mean field approximation of interacting Bosons. I. \emph{Commun.\ Math.\ Phys.}, 324(2):601--636 (2013).

\bibitem{GM2017} M.~Grillakis and M.~Machedon. Pair excitations and the mean field approximation of interacting Bosons. II. \emph{Commun. PDE} 42(2), 24--67 (2017).

\bibitem{GMM2010} M.~Grillakis, M.~Machedon and D.~Margetis. Second-order corrections to mean field evolution of weakly interacting bosons. I. \emph{Commun. Math. Phys.} 294(1), 273 (2010).

\bibitem{GMM2011} M.~Grillakis, M.~Machedon and D.~Margetis. Second-order corrections to mean field evolution of weakly interacting bosons. II. \emph{Adv. Math.} 228(3), 1778--1815 (2011).

\bibitem{hiroshima1998}
F.~Hiroshima.
Weak coupling limit with a removal of an ultraviolet cutoff for a Hamiltonian of particles interacting with a massive scalar field.
\emph{Infin. Dimens. Anal. Qu.} 1, 407--423 (1998).

\bibitem{KnowlesP2009} A.~Knowles and P. Pickl. Mean-field dynamics: singular potentials and rate of convergence. \emph{Commun. Math. Phys.} 298, 101--138 (2009).

\bibitem{Kuz2017} E.~Kuz. Exact evolution versus mean field second-order correction for bosons interacting via short-range two-body potential. \emph{Differ. Integral Equ.} 30(7/8), 587--630 (2017).

\bibitem{LSTT2018} J.~Lampart, J.~Schmidt, S.~Teufel and R.~Tumulka. Particle creation at a point source by means of interior-boundary conditions. \emph{Math. Phys. Anal. Geom.} 21, 12 (2018).

\bibitem{LMRSS2021}
N.~Leopold, D.~Mitrouskas, S.~Rademacher, B.~Schlein and R.~Seiringer. Landau--Pekar equations and quantum fluctuations for the dynamics of a strongly coupled polaron. \emph{Pure Appl. Anal} (in press).

\bibitem{LMS2021} N.~Leopold, D.~Mitrouskas and R.~Seiringer. Derivation of the Landau--Pekar equations in a many-body mean-field limit. \emph{Arch. Ration. Mech. Anal.} 240, 383--417 (2021).

\bibitem{LP2019} N.~Leopold and S.~Petrat. Mean-field dynamics for the Nelson model with fermions. \emph{Ann. Henri Poincar{\'e}} 20(10), 3471--3508 (2019).

\bibitem{LP2018}
N.~Leopold and P.~Pickl.
Mean-field limits of particles in interaction with quantized radiation fields.
In: D.~Cadamuro, M.~Duell, W.~Dybalski, and S.~Simonella (eds) \emph{Macroscopic Limits of Quantum Systems}, volume 270 of Springer Proceedings in Mathematics \& Statistics, 185--214 (2018).

\bibitem{LP2020}
N.~Leopold and P.~Pickl.
Derivation of the Maxwell--Schr\"odinger equations from the Pauli--Fierz Hamiltonian. 
\emph{SIAM J. Math. Anal.} 52(5), 4900--4936 (2020).

\bibitem{LRSS2019}
N.~Leopold, S.~Rademacher, B.~Schlein and R.~Seiringer.
The Landau--Pekar equations: adiabatic theorem and accuracy.
\emph{ Anal. \& PDE} (in press).

\bibitem{Lewin:2015a} M.~Lewin, P.\,T.~Nam, and B.~Schlein. Fluctuations around Hartree states in the mean-field regime. \emph{Am.\ J.\ Math.}, 137(6):1613--1650 (2015).

\bibitem{LNSS2015} 
M.~Lewin, P.\,T.~Nam, S.~Serfaty, J.\,P.~Solovej. Bogoliubov spectrum of interacting Bose gases. \emph{Comm. Pure Appl. Math.} 68(3), 413--471 (2015).

\bibitem{M2006} O.~Matte and J.S.~M\o ller. Feynman-Kac formulas for the
  ultra-violet renormalized Nelson model. \emph{Ast\'erisque} 404 (2018).

\bibitem{MNO2019} A.~Michelangeli, P.\,T. Nam and A.~Olgiati. Ground state energy of mixture of Bose gases. \emph{Rev.\ Math.\ Phys.}, 31(2) (2019).

\bibitem{M2021} D.~Mitrouskas. A note on the Fr\"ohlich dynamics in the strong coupling limit. \emph{Lett. Math. Phys.} 111, 45 (2021). 

\bibitem{mpp} D.~Mitrouskas, S.~Petrat, and P.~Pickl. Bogoliubov corrections and trace norm convergence for the Hartree dynamics. \emph{Rev.\ Math.\ Phys.}, 31(8) (2019).   

\bibitem{M2006_2} J.S.~M\o ller. On the essential spectrum of the
  translation invariant Nelson model. \emph{Mathematical physics of quantum
    mechanics}, 179–195. \emph{Lecture Notes in Physics}, 690, Springer,
  Berlin (2006).

\bibitem{nam:2016} P.\,T.~Nam and M.~Napi{\'o}rkowski. A note on the validity of Bogoliubov correction to mean-field dynamics. \emph{J.\ Math.\ Pures Appl.}, 108(5):662--688 (2017).

\bibitem{nam:2015} P.\,T.~Nam and M.~Napi{\'o}rkowski. Bogoliubov correction to the mean-field dynamics of interacting bosons. \emph{Adv.\ Theor.\ Math.\ Phys.}, 21(3):683--738 (2017).

\bibitem{namnap_review} P.\,T.~Nam and M.~Napi{\'o}rkowski. Norm approximation for many-body quantum dynamics and Bogoliubov theory. In: A.\ Michelangeli, and G. Dell'Antonio, editors, \emph{Advances in Quantum Mechanics: contemporary trends and open problems}. Springer-INdAM Series, 18 (2017).

\bibitem{namnap_low_dim} P.\,T.~Nam and M.~Napi{\'o}rkowski. Norm approximation for many-body quantum dynamics: focusing case in low dimensions.  \emph{Adv. Math.} 350, 547--587 (2019).

\bibitem{nelson}
E.~Nelson. Interaction of nonrelativistic particles with a quantized scalar field. 
\emph{J. Math. Phys.} 5, 1190 (1964).

\bibitem{pecher}
H.~Pecher.
Some new well-posedness results for the {K}lein--{G}ordon--{S}chr\"odinger system.
\emph{Diff. Int. Equations} 25(1/2), 117--142 (2012).

\bibitem{soffer} S.~Petrat, P.~Pickl, and A.~Soffer. Derivation of the Bogoliubov time evolution for a large volume mean-field limit. \emph{Ann. Henri Poincar{\'e}} 21(2), 461--498 (2020).

\bibitem{pickl}
P.~Pickl.
A simple derivation of mean field limits for quantum systems.
\emph{Lett. Math. Phys.} 97, 151--164 (2011).

\bibitem{P2005}
  A.~Pizzo.
  Scattering of an \emph{Infraparticle}: the one particle sector in Nelson's
  massless model.
  \emph{Ann. Henri Poincar\'e} 6(3), 553-606 (2005).
    
\bibitem{RS2009}
I.~Rodnianski and B.~Schlein.
Quantum fluctuations and rate of convergence towards mean field dynamics. \emph{Commun. Math. Phys.} 291(1), 31--61 (2009).    

\bibitem{JPS2007}
J.\,P.\ Solovej.
Many Body Quantum Mechanics. \emph{Lecture notes}, \url{http://web.math.ku.dk/~solovej/MANYBODY/mbnotes-ptn-5-3-14.pdf}, 2007.

\bibitem{teufel}
S.~Teufel. Effective N-body dynamics for the massless Nelson model and adiabatic decoupling without spectral gap. \emph{Ann. Henri Poincar\'e} 3,  939--965  (2002).

\end{thebibliography}
\end{document}